\documentclass[a4paper]{article}
\usepackage{theorem}
\usepackage{amsfonts,amsmath,amssymb,cite,mathrsfs}
\usepackage[dvips]{graphicx}
\newtheorem{theorem}{Theorem}
\newtheorem{lemma}{Lemma}
\newtheorem{corollary}[theorem]{Corollary}

\newcommand{\bC}{\mathbb{C}}
\newcommand{\bG}{\mathbb{G}}

\newcommand{\bJ}{\mathbb{J}}
\newcommand{\bK}{\mathbb{K}}

\newcommand{\bT}{\mathbb{T}}
\newcommand{\bZ}{\mathbb{Z}}
\newcommand{\sH}{\mathscr{H}}

\newcommand{\Vt}{ {V_{t,\kappa}} }
\newcommand{\Wt}{ {W_{t,\kappa}} }
\newcommand{\Xt}{ {X_t} }
\newcommand{\Xtr}{ {X_{t,r}} }
\newcommand{\Wtd}{ {W_{t,\kappa}^\prime} }
\newcommand{\Xtd}{ {X_t^*} }
\newcommand{\Kk}{ {\bK_\kappa} }

\newcommand{\Jk}{ {\bJ_\kappa} }
\newcommand{\tc}{\tilde{c}}

\newcommand{\lequiv}{=}
\newcommand{\requiv}{=}
\newcommand{\graph}{a graph with joined half lines}

\renewcommand{\Re}{{\rm Re}}

\newcommand{\ket}[1]{|{#1}\rangle}

\newcommand{\bra}[1]{\langle{#1}|}

\newcommand{\bi}[1]{\ensuremath{\boldsymbol{#1}}}
\newcommand{\er}[1]{\epsilon_{#1}}
\newcommand{\erd}[1]{ {\epsilon_{#1}^\prime} }

\newcommand{\sr}[1]{ {\sigma_{#1}} }

\newcommand{\hr}[1]{ {h_{#1} } }

\title{{\Large {\bf Limit theorems for the discrete-time quantum walk \\ on a graph with joined half lines  %title
}
}}

\author{ 
{\small 
Kota Chisaki,$^{1}$ %authorname1
\footnote{
chisaki-kota-jm@ynu.ac % email adress1
}\quad 
Norio Konno,$^{2}$ %authorname2
\footnote{konno@ynu.ac.jp % email adress2
}\quad  
Etsuo Segawa,$^{3}$ %authorname3
\footnote{To whom correspondence should be addressed. 
segawa@stat.t.u-tokyo.ac.jp % email adress3
}\quad  
}\\ 
{\scriptsize $^{1,2}$ 
Department of Applied Mathematics, Faculty of Engineering, Yokohama National University%shozoku
}\\
{\scriptsize Hodogaya, Yokohama 240-8501, Japan%juusho
} \\
{\scriptsize $^3$ 
Department of Mathematical Informatics, The University of Tokyo,%shozoku
}\\
{\scriptsize Bunkyo, Tokyo, 113-8656, Japan%juusho
}
} 

\date{\empty }
\begin{document}
\maketitle

\par\noindent
\begin{small}
\par\noindent
{\bf Abstract}. 
We consider a discrete-time quantum walk $W_{t,\kappa}$ at time $t$ on a graph with joined half
lines $\mathbb{J}_\kappa$, which is composed of $\kappa$ half lines with the same origin. Our analysis is based on
a reduction of the walk on a half line. The idea plays an important role to analyze the
walks on some class of graphs with \textit{symmetric} initial states. In this paper, we introduce
a quantum walk with an enlarged basis and show that $W_{t,\kappa}$ can be reduced to the walk
on a half line even if the initial state is \textit{asymmetric}. For $W_{t,\kappa}$, we obtain two types of
limit theorems. The first one is an asymptotic behavior of $W_{t,\kappa}$ which corresponds to
localization. For some conditions, we find that the asymptotic behavior oscillates. The
second one is the weak convergence theorem for $W_{t,\kappa}$. On each half line, $W_{t,\kappa}$ converges
to a density function like the case of the one-dimensional lattice with a scaling order of $t$. 
The results contain the cases of quantum walks starting from the general initial state 
on a half line with the general coin and homogeneous trees with the Grover coin.
%abst

\end{small}

\setcounter{equation}{0}

\section{Introduction}
\label{intro}
Random walks have a very important role in various fields, such as physical systems, mathematical
modeling and computer algorithms. In 1990s, quantum walks arise as a quantum
counterpart of random walks \cite{aharonov01,meyer01,farhi01}.
They are defined by unitary evolutions of probability
amplitudes, whereas random walks are obtained by evolutions of probabilities by transition
matrices. Discrete-time quantum walks are introduced by Refs. \cite{aharonov01,meyer01}.
In recent years, quantum walks have been well developed in fields of quantum algorithms, for example \cite{childs01,shenvi01,ambainis01}.
On the other hand, studies of the walks from the mathematical point of view also arise.
Especially, as a limiting behaver, localization appears in quantum cases \cite{inui03,chisaki01,konno04,shikano01,cantero01,cantero02}.
Furthermore the quantum walk has a quadratically faster scaling order than
the random walk in the weak convergence \cite{konno01,konno02,grimmett01,miyazaki01,segawa01}.
Cantero et al. introduced an analysis using the CMV matrix \cite{cantero01,cantero02}. This method is very useful to consider localization.
To analyze the quantum walk, we use the
generating function. By using the generating function, we can compute not only localization
but also the weak convergence of the walk. A reduction technique \cite{krovi01,tregenna01,carneiro01}, which reduces
the walk to a one-dimensional quantum walk, is very important to apply a path counting
method \cite{konno01,konno02,oka01} which gives an explicit expression for the generating function. To treat the
quantum walk with asymmetric initial states, we introduce a quantum walk with enlarged bases.

%\subsection{Summary of results} \label{sec:summary of results} 
Our main results are two limit theorems for the quantum walk $\Wt$ on \graph\ 
with \textit{arbitrary} initial state starting from the origin.
In case of $\kappa=1$, $W_{t,1}$ corresponds to a quantum walk on a half line with the general coin.
Furthermore, by considering the reduction of the walks, 
the two limit theorems can be adopted to quantum walks on homogeneous trees and semi-homogeneous trees with the Grover coin operator.
One of two our main results is the explicit expression for the limit probability of $\Wt$. 
It is corresponding to localization which is defined that there exists
a vertex of the graph $x$ such that $\limsup_{t\to\infty}P(\Wt=x)>0$.   
We find that, for some conditions, the asymptotic behavior oscillates.
Same as other results on quantum walks \cite{inui03,chisaki01,konno04}, localization has an exponential decay for position $x$ on each half line.
Another main result is the weak convergence of $\Wt$.
On each half line, $\Wt$ has a scaling order $t$.
Moreover the limit measure has a typical density
function which appears on other quantum walks \cite{chisaki01,konno01,konno02,grimmett01,miyazaki01,segawa01}.

For related works, Chisaki et al.~\cite{chisaki01} obtained the same type of limit theorems for a quantum walk on homogeneous trees with two special initial states.
This result induces limit theorems for a quantum walk on a half line with a special coin operator.
Konno and Segawa \cite{segawa02} showed localization of quantum walks on a half line by using the spectral analysis of the corresponding CMV matrices.

%\subsection{Organization of the paper} \label{sec:organization of the paper} 
The remainder of the present paper is organized as follows. In Section 2, we give definitions
of discrete-time quantum walks treated in this paper. Section 3 presents our results. Section
4 gives proofs of our main theorems. In Subsection 4.1, we introduce a quantum walk with
an enlarged basis and reduce $W_{t,\kappa}$ to the walk on a half line. Subsection 4.2 presents a
proof of Theorem 1 based on the generating function. Subsection 4.3 is devoted to a proof of
Theorem 2 using the Fourier transform of the generating function. In Appendix, we compute
the generating function.
%%%%%%%%%%%%%%%%%%%%%%%%%%%%%%%%%%%%%%%%%
%%%%%%%%%%%%%%%%%%%%%%%%%%%%%%%%%%%%%%%%%
\section{Discrete-time quantum walks} \label{sec:discrete-time quantum walks} 
This section gives the definition of the quantum walk on undirected connected graph $\bG$. 
Let $V(\bG)$ be a set of all vertices in $\bG$ and $E(\bG)$ be a set of all edges in $\bG$.
Here we define $E_x(\bG)\subset E(\bG)$ as a set of all edges which connect the vertex $x\in V(\bG)$. % and the degree of a site $x$ is written by $|E_x(\bG)|$.
Now we take a Hilbert space spanned by an orthonormal basis $\{ \ket{x};\ x\in V(\bG) \}$ as a position space $\sH_p$
and a Hilbert space generated by an orthonormal basis $\{ \ket{l};\ l\in E_x(\bG) \}$ for $x\in V(\bG)$ as a \textit{local} coin space $\sH_{c_x}$.
A discrete-time quantum walk on $\bG$ is defined on a Hilbert space $\sH$ spanned by an orthonormal basis $\{\ket{x,l};\ x\in V(\bG),\ l\in E_x(\bG) \}$.
Note that if we take $\bG$ as a regular graph, $\sH$ can be written as $\sH=\sH_p\otimes \sH_{c_x}$ for any $x\in V(G)$.
On the space $\sH$, the evolution operator $U$ is given by $U=SF$,
where $S:\sH\to\sH$ is a shift operator and $F:\sH\to\sH$ is a coin operator.
Here we define $F=\sum_{x\in V(\bG)}\ket{x}\bra{x}\otimes C_x$ as a coin operator and 
$C_x:\sH_{c_x} \to \sH_{c_x}$ for $x\in V(\bG)$ as a \textit{local} coin operator.
If the graph is regular and the local coin operator is all the same, we can rewrite the coin operator as $F=I_p\otimes C$, where $I_p$ is the identity operator on $\sH_p$.
%To generate probability distribution from the process, $U$ must be unitary operator on $\sH$.
%We find that if all the local coin operators are unitary and the shift operator only works as a permutation of amplitudes on the basis of $\sH$,
%$U=SF$ is an unitary operator.
As typical local coin operators, the Hadamard operator $H$ and the Grover operator $G_d$ are often used,
where $H$ and $G_d$ $(d\geq2)$ are defined by
\begin{align*}
H =& \frac{1}{\sqrt{2}}
\left[ 
\begin{array}{cc}
1&1\\
1&-1\\
\end{array}
\right] ,\\
G_d =&  
\left[ 
\begin{array}{cccc}
a_d&b_d&\cdots&b_d\\
b_d&a_d&\cdots&b_d\\
\vdots&\vdots&\ddots&\vdots\\
b_d&b_d&\cdots&a_d\\
\end{array}
\right]
\lequiv
\left[ 
\begin{array}{cccc}
\frac{2}{d}-1&\frac{2}{d}&\cdots&\frac{2}{d}\\
\frac{2}{d}&\frac{2}{d}-1&\cdots&\frac{2}{d}\\
\vdots&\vdots&\ddots&\vdots\\
\frac{2}{d}&\frac{2}{d}&\cdots&\frac{2}{d}-1\\
\end{array}
\right].
\end{align*}
In this paper we define $a_1=1$, $b_1=2$ and $G_1=1$.
From the construction, the state at time $t$ and position $x$ is described as
\begin{eqnarray}
\Psi_t(x) &=& \sum_{l\in E_x(\bG)} \alpha_t(x,l) \ket{x,l} , \label{eq:alpha}
\end{eqnarray}
where $\alpha_t(x,l)\in\bC$ is the amplitude of the base $\ket{x,l}$ at time $t$ and $\bC$ is the set of all complex numbers.
The probability of the state is given by a square norm of $\Psi_t(x)$, i.e., $\|\Psi_t(x)\|^2 = \sum_{l\in E_x(\bG)}|\alpha_t(x,l)|^2$.
We only consider the initial state starting from the origin ``$o$'' with the state $\Psi_0(o)$ such that $\|\Psi_0(o)\|=1$. 
%%%%%%%%%%%%%%%%%%%%%%%%%%%%%%%%%%%%%%%%%
\subsection{Quantum walk on \graph} \label{sec:quantum walk on \graph}
This subsection gives the definition of \graph\ $\Jk$ and the quantum walk $\Wt$ on $\Jk$.
Let $\Kk = \{0,1,\ldots,\kappa-1\}$ and $\bZ_r \lequiv \{\hr{r}(1),\hr{r}(2),\ldots\}$ for $r\in\Kk$, 
we define $V(\mathbb{J}_\kappa)=\{0\}\cup \{\cup_{j\in \mathbb{K}_\kappa}\mathbb{Z}_j\}$. 
A vertex $h_i(x)$ connects $h_j(y)$ if and only if $|x-y|=1$ with $i=j$, and the origin $0$ connects $h_r(1)$ for any $r$ 
(see Fig. \ref{fig:bhlmat} (a) for example).

The quantum walk on $\Jk$ is defined on $\sH^{(\Jk)}$ which is a Hilbert space spanned by an orthonormal basis 
$\{ \ket{0,l};\ l\in\{ \er{0},\er{1},\ldots,\er{\kappa-1} \} \}\cup\{\ket{x,l};\ x\in V(\Jk)\setminus \{0\},\ l\in\{Up,Down\} \}$.
Throughout this paper, we put the base $\ket{\er{r}}$ as 
${}^T[ \overbrace{0 \cdots 0}^{r} 1 \overbrace{0 \cdots 0}^{\kappa-r-1} ],$
where $T$ is the transposed operator. 
We define a local coin operator $C$ as 
\begin{equation} C=\begin{bmatrix} a & b \\ c & d \end{bmatrix}\in U(2)\mathrm{\;with}\;abcd\neq 0, \end{equation}
where $U(d)$ is the set of $d\times d$ unitary matrices. 
The coin operator $F_J$ is given by 
\begin{eqnarray}\label{mon}
F_J \lequiv \ket{0}\bra{0} \otimes G_\kappa + \sum_{x\in V(\Jk)\setminus \{0\}} \ket{x}\bra{x} \otimes C , \label{eq:FBk}
\end{eqnarray}
where $G_\kappa$ is the Grover operator.
The shift operator $S_J$ is given by
\begin{eqnarray*}
S_J\ket{0,l} &=& \ket{\hr{r}(1),Down} \ ,\ l = \er{r} ,\\
S_J\ket{\hr{r}(1),l} &=& \left\{ 
\begin{array}{ll}
\ket{0,\er{r}} , & l = Up ,\\
\ket{\hr{r}(2),Down} , & l = Down ,\\
\end{array} \right. \\
S_J\ket{\hr{r}(x),l} &=& \left\{ 
\begin{array}{ll}
\ket{\hr{r}(x-1),Up} , & l = Up ,\\
\ket{\hr{r}(x+1),Down} , & l = Down ,\\
\end{array} \right.  \  x\geq 2 .
\end{eqnarray*}
%We can define the local coin operator at the origin as $U(\kappa)$, however we take it as the Grover operator added a complex phase owing to the convenience of the induction.
Then the evolution operator of the walk $U_J$ is obtained by $U_J=S_JF_J$. 
An expression of $W_{t,3}$ using weights is shown in Fig. \ref{fig:bhlmat} (a), where
\begin{eqnarray*}
P_1^{\er{1}} = 
\left[ 
\begin{array}{cc}
a&b\\
0&0\\
0&0\\
\end{array}
\right], && 
Q_0^{\er{1}} = 
\left[ 
\begin{array}{ccc}
0&0&0\\
-\frac{1}{3}&\frac{2}{3}&\frac{2}{3}\\
\end{array}
\right] , \\
P_1^{\er{2}} =
\left[ 
\begin{array}{cc}
0&0\\
a&b\\
0&0\\
\end{array}
\right], && 
Q_0^{\er{2}} = 
\left[ 
\begin{array}{ccc}
0&0&0\\
\frac{2}{3}&-\frac{1}{3}&\frac{2}{3}\\
\end{array}
\right] , \\
P_1^{\er{3}} =
\left[ 
\begin{array}{cc}
0&0\\
0&0\\
a&b\\
\end{array}
\right], && 
Q_0^{\er{3}} = 
\left[ 
\begin{array}{ccc}
0&0&0\\
\frac{2}{3}&\frac{2}{3}&-\frac{1}{3}\\
\end{array}
\right] , \\
P =
\left[ 
\begin{array}{cc}
a&b\\
0&0\\
\end{array}
\right], && 
Q =
\left[ 
\begin{array}{cc}
0&0\\
c&d\\
\end{array}
\right] .
\end{eqnarray*}
Note that $W_{t,1}$ is a quantum walk on a half line with a reflecting wall.
\begin{figure}[h]
  \begin{center}
    \begin{tabular}{ll}
      (a)&\ \ \ \ (b)\\
      \resizebox{0.3\textwidth}{!}{\includegraphics{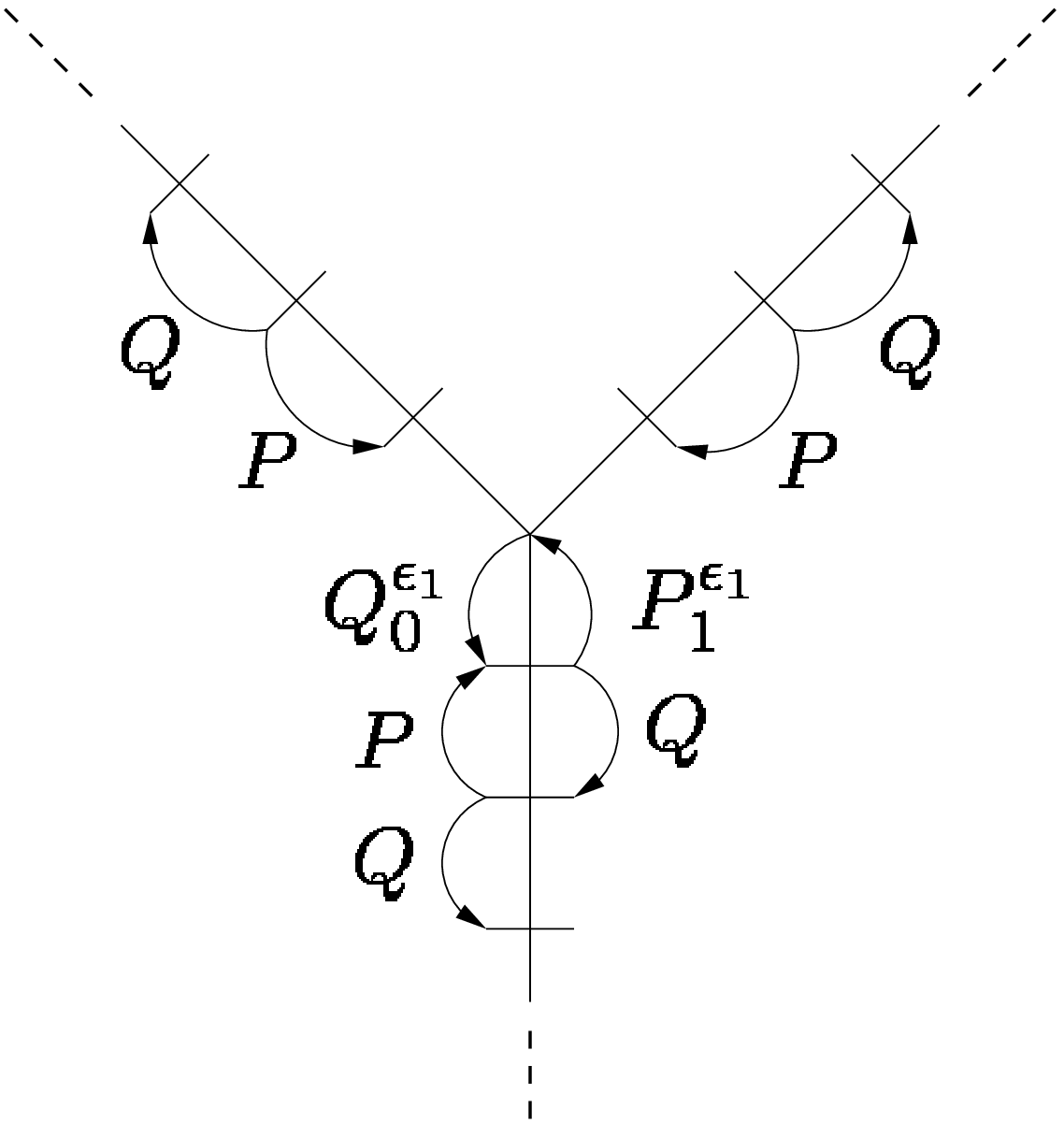}}&
      \ \ \ \ \resizebox{0.3\textwidth}{!}{\includegraphics{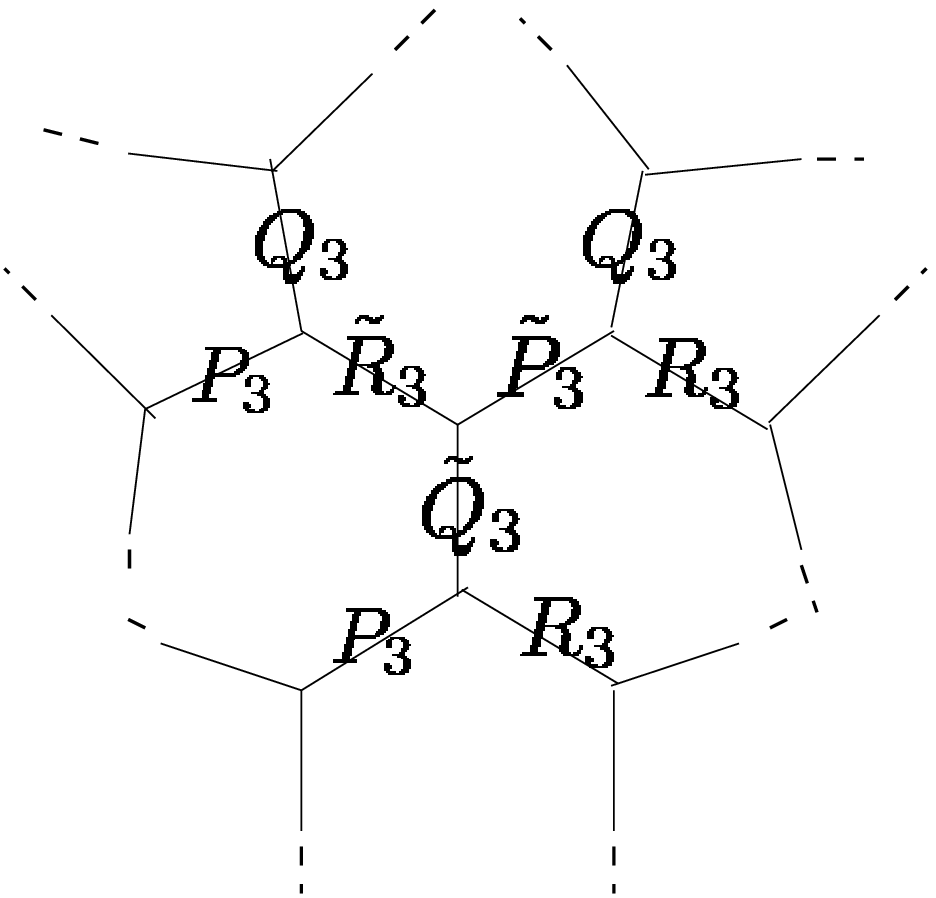}}
    \end{tabular}
  \caption{(a) Quantum walk on $\bJ_3$, (b) Quantum walk on $\mathbb{T}_3$}
  \label{fig:bhlmat}
  \end{center}
\end{figure}
%
%%%%%%%%%%%%%%%%%%%%%%%%%%%%%%%%%%%%%%%%%
\subsection{Quantum walk on homogeneous trees} \label{sec:quantum walk on homogeneous trees}
We define a homogeneous tree $\bT_\kappa$ and a quantum walk $\Vt$ on $\bT_\kappa$.
%The definition of a quantum walk on $\bT_\kappa$ is given by our previous paper \cite{chisaki01}.
Fix $\kappa \geq 2$, let $\Sigma \lequiv \{\sr{0},\sr{1},\ldots,\sr{\kappa-1}\}$ be the set of generators subjected to the relation
$\sr{j}^2 = e$ for $j\in\Kk$, where the empty word $e$ is the unit of this group. 
Then we put 
$V(\bT_\kappa) \lequiv \{e\}\cup\{\sr{i_n}\ldots \sr{i_2} \sr{i_1} : n\geq 1,\ \sr{i_j}\in \Sigma, \ i_{j+1} \neq i_j \ \mathrm{for}\ j=1,2,\ldots,n-1 \}$. 
Here vertices $g$ and $h$ are connected if and only if $gh^{-1} \in \Sigma$.
On this graph, $\sH_p^{(\bT_\kappa)}$ is generated by an orthonormal basis $\{\ket{g};\ g\in V(\bT_\kappa) \}$ 
and $\sH_c^{(\bT_\kappa)}$ is associated with an orthonormal basis $\{\ket{\sr{j}};\ \sr{j} \in \Sigma \}$.
We choose $G_\kappa$ as the local coin operator, then the coin operator
$F_T$ and the shift operator $S_T$ are defined as follows: for $\sigma\in\Sigma$
\begin{eqnarray*}
&& F_T \lequiv \ket{e}\bra{e} \otimes \tc G_\kappa + \sum_{g\in V(\bT_\kappa)\setminus \{e\}} \ket{g}\bra{g} \otimes G_\kappa ,\\
&& S_T\ket{g,\sigma} = \ket{\sigma g,\sigma} ,
\end{eqnarray*}
where we put $\ket{\sr{r}}$ as $ {}^T[ \overbrace{0 \cdots 0}^{r} 1 \overbrace{0 \cdots 0}^{\kappa-r-1} ]$ and $\tc\in \mathbb{C}$ with $|\tc|=1$. 
The phase $\tc$ works as a defect on the origin, which is an extension of our model in \cite{chisaki01}. 
An expression of $V_{t,3}$ using weights is shown in Fig. \ref{fig:bhlmat} (b), where
\begin{eqnarray*}
&& P_3 = 
\left[ 
\begin{array}{ccc}
-\frac{1}{3}&\frac{2}{3}&\frac{2}{3}\\
0&0&0\\
0&0&0\\
\end{array}
\right], \ 
Q_3 = 
\left[ 
\begin{array}{ccc}
0&0&0\\
\frac{2}{3}&-\frac{1}{3}&\frac{2}{3}\\
0&0&0\\
\end{array}
\right], \ 
R_3 = 
\left[ 
\begin{array}{ccc}
0&0&0\\
0&0&0\\
\frac{2}{3}&\frac{2}{3}&-\frac{1}{3}\\
\end{array}
\right] ,
\end{eqnarray*}
and $\tilde{P}_3=\tc P_3$, $\tilde{Q}_3=\tc Q_3$, $\tilde{R}_3=\tc R_3$. 

In the case of the one point initial state on the origin, $V_{t,\kappa}$ can be reduced to the equivalent
walk on $\mathbb{J}_\kappa$ even if the initial state is not symmetric. 
To explain it, we define subgraph $\mathbb{T}_{\kappa}^{(r)}\in \mathbb{T}_\kappa$ as 
$V (\mathbb{T}_\kappa^{(r)}) 
= \{\sigma_{i_n}\cdots \sigma_{i_2}\sigma_{i_1}: 
n\geq 1,\;\sigma_{i_j}\in \Sigma,\;\sigma_{i_1} = \sigma_{r},\;i_{j+1} \neq i_j\;\mathrm{for}\;j = 1,2,\dots,n-1 \}$. 
They are subtrees whose roots are the children of the root of $\mathbb{T}_\kappa$. 
Now we consider the following new basis, for $x\geq 1$,
\begin{align*}
|x,Up\rangle_{\sigma_r} &= \frac{1}{\sqrt{(\kappa-1)^{x-1}}}
	\sum_{\scriptsize{\begin{matrix}g\in V(\mathbb{T}_\kappa^{(r)}) \\ |g|=x\end{matrix}}} \sum_{\sigma_j:|\sigma_jg|=x-1} |g,\sigma_j\rangle, \\
|x,Down\rangle_{\sigma_r} &= \frac{1}{\sqrt{(\kappa-1)^x}}
	\sum_{\scriptsize{\begin{matrix}g\in V(\mathbb{T}_\kappa^{(r)}) \\ |g|=x\end{matrix}}} \sum_{\sigma_j:|\sigma_jg|=x+1} |g,\sigma_j\rangle.
\end{align*}
The new space ${\sH^{(\mathbb{T}_\kappa)}}'$ spanned by a basis 
$\{|e,l\rangle:l\in \Sigma\}\cup \{|x,l\rangle_{\sigma_r}: r\in \mathbb{K}_\kappa,x\in \mathbb{Z}_+,l\in \{Up,Down\}\}$ is isomorphic to 
${\sH^{(\mathbb{J}_\kappa)}}$ under the following one-to-one correspondence
\begin{align}
|x,l\rangle_{\sigma_r} &\leftrightarrow  |h_r(x),l\rangle\;\mathrm{for}\;l\in\{Up,Down\},\;x\geq 1, \notag  \\
|e,\sigma_r\rangle    &\leftrightarrow  |0,\epsilon_r\rangle, \label{nagoya}
\end{align}
where $\mathbb{Z}_+=\{1,2,\dots\}$. Then the direct computation gives the following lemma. 
%%%%%%%%%%%%%%%%%
\begin{lemma}[Homogeneous tree] \label{lem:homogeneous tree}
The subspace of ${\sH^{(\mathbb{T}_\kappa)}}'$ is invariant under the action of the time evolution of $V_{t,\kappa}$. 
In particular, when we take the bijection from ${\sH^{(\mathbb{T}_\kappa)}}'$ to $\sH^{(\mathbb{J}_\kappa)}$ 
given by Eq. (\ref{nagoya}), the walk is equivalent to $W_{t,\kappa}$ with the following coin operator
\begin{eqnarray*}
F_J^{(\mathbb{T}_\kappa)} = \ket{0}\bra{0} \otimes \tc G_\kappa + \sum_{x\in V(\Jk)\setminus \{0\}} \ket{x}\bra{x} 
	\otimes 
\left[ 
\begin{array}{cc}
a_\kappa&\sqrt{\kappa-1}b_\kappa\\
\sqrt{\kappa-1}b_\kappa& -a_\kappa\\
\end{array}
\right] . 
\end{eqnarray*}
\end{lemma}
When we consider $\tc^{-1}F_J^{(\mathbb{T}_\kappa)}$, the above equation becomes a special case of Eq. (\ref{mon}), since $|\tc|=1$. 

%We give an outline of the proof.
%Let $\bT^{(0)}_\kappa,\ldots,\bT^{(\kappa-1)}_\kappa$ be subtrees of $\bT_\kappa$ whose roots are sites $\sr{0},\ \ldots,\ \sr{\kappa-1}$, respectively.
%Lemma 1 in \cite{chisaki01} means that for any $r\in\Kk$, the sites in $\bT^{(r)}_\kappa$ with the same depth from the origin have the same probability even if the initial state is asymmetric.
%Therefore we can reduce each $\bT^{(r)}_\kappa$ to a half line by using the same method in \cite{chisaki01}.
%In this process, $c/\tc=-a_\kappa$ so we can adopt corollary \ref{cor:phi=0}.
%Note that $C_\kappa$ corresponds to $H_\kappa$ in \cite{chisaki01}.
%A correspondence between $\bT_3$ and $\bJ_3$ is shown in Fig. \ref{fig:htree}.

Similar to $\Vt$, we can define a quantum walk $V_{t,\kappa',\kappa}$ on a semi-homogeneous tree $\mathbb{T}_{\kappa',\kappa}$, 
which is a $\kappa$-regular tree except the origin whose degree is $\kappa^\prime\geq 2$,
with the local coin operator $\tc G_{\kappa^\prime}$ at the origin and $G_\kappa$ otherwise.
Then we can reduce it to $W_{t,\kappa^\prime}$ with the coin operator $F_J^{(\mathbb{T}_{\kappa',\kappa})}$ given by
\begin{eqnarray*}
F_J^{(\mathbb{T}_{\kappa',\kappa})} 
= \ket{0}\bra{0} \otimes \tc G_{\kappa^\prime} + \sum_{x\in V(\bJ_{\kappa'})\setminus \{0\}} \ket{x}\bra{x} 
\otimes 
\left[ 
\begin{array}{cc}
a_\kappa&\sqrt{\kappa-1}b_\kappa\\
\sqrt{\kappa-1}b_\kappa& -a_\kappa\\
\end{array}
\right]. 
\end{eqnarray*}
The infinite binary tree is a special case for this graph ($\kappa=3,\ \kappa^\prime=2$).
%%%%%%%%%%%%%%%%%%%%%%%%%%%%%%%%%%%%%%%%%
%%%%%%%%%%%%%%%%%%%%%%%%%%%%%%%%%%%%%%%%%
\section{Main results}
In our main theorems, we give explicit formulae with respect to each half line in $\mathbb{J}_\kappa$. 
Let $\Psi_t(x)$ be the state of the quantum walk $W_{t,\kappa}$ at time $t$ and position $x$. 
For $x\in \mathbb{Z}_+\cup\{0\}$, we introduce random variables $X_{t,r}$ as $P(X_{t,r}=0)=|\alpha_t(0,\epsilon_r)|^2$ and 
$P(X_{t,r}=0)=||\Psi_t(h_r(x))||^2$. Remark that $P(W_{t,\kappa}=h_r(x))=P(X_{t,r}=x)$ for $x\geq 1$, 
$P(W_{t,\kappa}=0)=\sum_{j\in \mathbb{K}_\kappa}P(X_{t,j}=0)$ and $\sum_{j\in\mathbb{K}_\kappa}\sum_{x\in \mathbb{Z}_+\cup \{0\}}P(X_{t,j}=x)=1$. 

In order to describe the limit theorems for $W_{t,\kappa}$, we first introduce several parameters.
\begin{eqnarray*}
&& \phi \lequiv \arg{ (c) },\\
&& K_{\pm} \lequiv \left| 1 \pm c \right|^2,\\
&& K_{\times} \lequiv (1-c)(1+\bar{c}),
\end{eqnarray*}
where $\bar{a}$ is the complex conjugation of $a\in \mathbb{C}$. 
Next, we denote the following notations to state Theorem \ref{thm:localization}. 
Localization is described by three terms $L_m(x)$, $L^r_p(x)$ and $L^r_c(x)$. 
\begin{align*}
&\Psi_0(0) \requiv \sum_{j\in\Kk} \psi_j \ket{0,\er{j}}, \nonumber \\
&L_m(x) = \Gamma_-(x) \left|\sum_{j\in\mathbb{K}_\kappa}\psi_j\right|^2, \;\;
L_p^r(x) = \Gamma_+(x) \left|\sum_{j\in\mathbb{K}_\kappa}(\psi_j-\psi_r)\right|^2 ,\\
&L_c^r(x,t) =2b_\kappa^2\mathrm{Re}\left[ \Gamma_\times (x,t)
	\left(\sum_{j\in \mathbb{K}_\kappa}\overline{\psi_j}\right)\left(\sum_{j\in \mathbb{K}_\kappa}(\psi_j-\psi_r)\right) \right], \\
&\Gamma_\pm(x) = \frac{b_\kappa^2 |c|^2(\cos{\phi} \pm |c|)^2}{K_\pm^2} 
	\left\{ \delta_0(x)+(1-\delta_0(x))\left(\frac{|a|^2}{K_\pm}\right)^{x-1}\left(1+\frac{|a|^2}{K_\pm}\right) \right\},\\
&\Gamma_\times(x,t) = \left( \sqrt{ \frac{K_\times}{\ \overline{K_\times}\ } }\right)^{t+1} \frac{|c|^2({\cos^2{\phi}}-|c|^2)}{K_\times^2} \\ 
&\;\;\;\;\;\;\;\;\;\;\;\;\;\;\;\;\times\left\{ -\delta_0(x) \sqrt{\frac{K_\times}{\ \overline{K_\times}\ }} + (1-\delta_0(x)) \left( \frac{|a|^2}{\sqrt{K_+K_-}} \right)^{x-1} 
				\left( 1-\frac{|a|^2}{\ \overline{K_\times}\ } \right)  \right\} .
\end{align*}
Then we have the following theorem. 
%%%%%%%%%%%%%%%%%%%%%%%%%%%%%%%%%%%%%%%%%%%%%%%%%%%%%%%
\begin{theorem}[Localization] \label{thm:localization}
For $\kappa\geq1$, $x\in\bZ_+\cup\{0\}$, $r\in\Kk$,
\begin{multline*}
P(\Xtr=x) \sim \\ 
\frac{1+(-1)^{t+x}}{2}  \left\{ I_{[-1,|c|)}\left(\cos{\phi}\right)L_m(x) + I_{(-|c|,1]}\left(\cos{\phi}\right)L_p^r(x) + I_{(-|c|,|c|)}\left(\cos{\phi}\right)L_c^r(x,t) \right\} ,
\end{multline*}
where $f(t)\sim g(t)$ means $f(t)/g(t)\to 1$  $(t\to \infty)$.
\end{theorem}
%%%%%%%%%%%%%%%%%%%%%%%%%%%%%%%%%%%%%%%%%%%%%%%%%%%%%%%
We see that in many cases the quantum walk on $\mathbb{J}_\kappa$ exhibits localization. Localization does
not occur only in the following two cases, ``$\sum_{j}(\psi_j-\psi_r)=0$ for any $r$ and $\cos\phi\geq |c|$" 
and ``$\sum_{j}\psi_j=0$ and $\cos\phi\leq -|c|$". 
Only the symmetric initial state (i.e., $\psi_i=1/\sqrt{\kappa}$ for any $i$) satisfies the first condition. 
Moreover $L^r_c(x, t)$ is an oscillatory term, so the probability oscillates if $L^r_c(x, t)$ exists. 
The probability $P(X_{t,0} = 1)$ is shown in Fig. 2, where we choose the local coin operator as $e^{i\varphi}H$. 
%%%%%%%%
\begin{figure}[htbp]
\begin{center}
	\includegraphics[width=115mm]{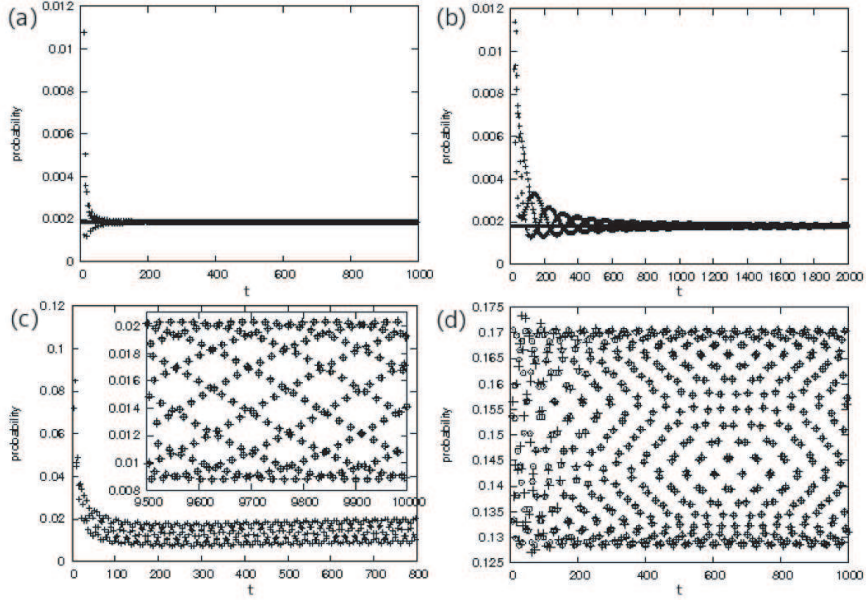}
\end{center}
\caption{Comparison between theoretical values with circles and numerical estimations with crosses of $W_{t,3}$ with $C=e^{i\varphi}H$.
The probability $P(X_{t,0}=1)$ is plotted. The initial state is $\psi_0=e^{i10\pi/180}/\sqrt{3}$, $\psi_1=e^{i30\pi/180}/\sqrt{3}$, $\psi_2=e^{i340\pi/180}/\sqrt{3}$ 
and (a) $\varphi=0$, (b) $\varphi=40\pi/180$, (c) $\varphi=50\pi/180$, (d) $\varphi=80\pi/180$. 
Since $|c|=1/\sqrt{2}$, $\varphi=45\pi/180$ is a critical point for the oscillatory behavior. 
In the large figure of (c), theoretical value are omitted. 
 }
\label{fig:one}
\end{figure}
%%%%%%%%
From Theorem 1, the condition for the existence $L^r_c(x, t)$ is $-|c| < \cos\phi < |c|$. 
Therefore, in this case, the oscillation emerges when $\pi/4 < \varphi < 3\pi/4$.
Remark that from Theorem 1 we can see the following relation, 
\[ \sum_{r\in\mathbb{K}_\kappa}L_c^r(x,t) 
	=2b_\kappa^2\mathrm{Re}\left[ \Gamma_\times (x,t)
	\left(\sum_{j\in \mathbb{K}_\kappa}\overline{\psi_j}\right)\left(\sum_{r\in\mathbb{K}_\kappa}\sum_{j\in \mathbb{K}_\kappa}(\psi_j-\psi_r)\right) \right]=0, \]
This means that the oscillation disappears when we take the probability summed over all
vertices with a same distance from the origin. In addition, since $P(W_{t,\kappa} = 0) =\sum_{j} P(X_{t,j} =0)$, 
the probability of the origin does not oscillate for any condition. We also find that the
distribution has an exponentially decay with $x$ from Theorem 1. The probability $P(X_{10000,0} =x)$ is shown in Fig. 3.
%%%%%%%%
\begin{figure}[htbp]
\begin{center}
	\includegraphics[width=115mm]{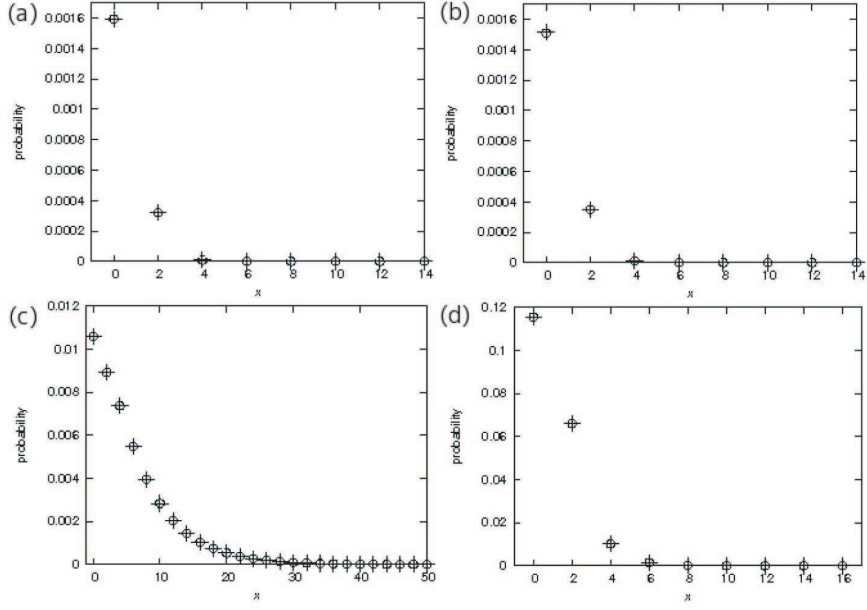}
\end{center}
\caption{Comparison between theoretical values with circles and numerical estimations with crosses of $W_{t,3}$ with $C=e^{i\varphi}H$.
The probability $P(X_{10000,0}=x)$ is plotted. The initial state is $\psi_0=e^{i10\pi/180}/\sqrt{3}$, $\psi_1=e^{i30\pi/180}/\sqrt{3}$, 
$\psi_2=e^{i340\pi/180}/\sqrt{3}$ 
and (a) $\varphi=0$, (b) $\varphi=40\pi/180$, (c) $\varphi=50\pi/180$, (d) $\varphi=80\pi/180$. 
}
\label{fig:one}
\end{figure}
%%%%%%%%

In order to state the weak convergence theorem, at first we define some parameters depending
on the initial state. For $r \in \mathbb{K}_\kappa$, we put 
\begin{align}
& \theta^r_1(\bi{\psi}) \lequiv \left| \psi_r \right|^2, \label{eq:theta1}\\
& \theta^r_2(\bi{\psi}) \lequiv \sum_{j\in\Kk\setminus \{r\}} \overline{\psi_r}\psi_j, \label{eq:theta2}\\
& \theta^r_3(\bi{\psi}) \lequiv \sum_{j\in\Kk\setminus \{r\}} \left| \psi_j\right|^2 +
 	\sum_{j,k\in\Kk\setminus \{r\} \atop j\neq k} \left( \psi_j\overline{\psi_k} 
	+\overline{\psi_j}\psi_k\right) \label{eq:theta3}.
\end{align}
\noindent 
Next, we introduce the following notations. Terms $C_m$ and $C^r_d$ are delta measures which are
caused by localization and $C^r_d(x)$ is a weight on density function $f_K(x)$ which is formed a
typical shape of one dimensional quantum walks. 
\begin{eqnarray*}
&& C_m = \frac{b_\kappa^2|c|(|c|-\cos{\phi})}{2K_-} \left| \sum_{j\in \mathbb{K}_\kappa}\psi_j\right|^2 ,\;\;
   C_p^r = \frac{b_\kappa^2|c|(|c|+\cos{\phi})}{2K_+} \left| \sum_{j\in \mathbb{K}_\kappa}(\psi_j-\psi_r)\right|^2,\\
&& C_d^r(x) = \frac{\Gamma_1(x) \theta^r_1(\bi{\psi}) + 2\Re(\Gamma_2(x) \theta^r_2(\bi{\psi})) + \Gamma_3(x) \theta^r_3(\bi{\psi})}{(K_+-(1-x^2)\sin^2{\phi})(K_--(1-x^2)\sin^2{\phi})}x^2 , \\
\end{eqnarray*}
\begin{eqnarray*}
&& \Gamma_1(x) = 4a_\kappa |c|(|a|^2-x^2)\cos{\phi}\sin^2{\phi}  \\
&&\;\;\;\;\;\;\;\;\;\;\;\;\;\;\;\;+ \left( a_\kappa^2 + 2a_\kappa |c|\cos{\phi} + 1  \right) \left( 1+|c|^2-2|c|^2\cos^2{\phi}-(1-x^2)\sin^2{\phi} \right) ,\\
&& \Gamma_2(x) = - 2b_\kappa|c|(|a|^2-x^2)ie^{i\phi}\cos{\phi}\sin{\phi} \\
&&\;\;\;\;\;\;\;\;\;\;\;\;\;\;\;\;+ b_\kappa (a_\kappa + |c|e^{i\phi})(1+|c|^2-2|c|^2\cos^2{\phi} - (1-x^2)\sin^2{\phi}), \\
&& \Gamma_3(x) = b_\kappa^2 (1+|c|^2-2|c|^2\cos^2{\phi}-(1-x^2)\sin^2{\phi}). \\
\end{eqnarray*}
%%%%%%%%%%%%%%%%%%%%%%
The weak convergence theorem is derived as follows. 
\begin{theorem}[Weak convergence] \label{thm:weak convergence}
For $\kappa\geq 1$, $r\in\Kk$, as $t\to\infty$ 
\begin{eqnarray*}
P \left( u\leq \frac{\Xtr}{t} \leq v \right) \to \int_u^v \rho_W^r(x)  dx.
\end{eqnarray*}
The limit measure is defined by 
\begin{eqnarray*}
\rho_W^r(x) = \left\{ I_{[-1,|c|)}\left(\cos{\phi}\right)C_m + I_{(-|c|,1]}\left(\cos{\phi}\right)C_p^r \right\} \delta_0(x) + C_d^r(x) f_K(x) ,
\end{eqnarray*}
where
\begin{eqnarray*}
 f_K(x) = \frac{I_{[0,a)}(x)\sqrt{1-|a|^2}}{\pi(1-x^2)\sqrt{|a|^2-x^2}}.
\end{eqnarray*}
\end{theorem}
%%%%%%%%%%%%%%%%%%%%%
The function $C_d^0(x)f_K(x)$ and the scaled numerical values are shown in Fig. \ref{fig:weak convergence}.
%%%%%%%%
\begin{figure}[t]
\begin{center}
	\includegraphics[width=115mm]{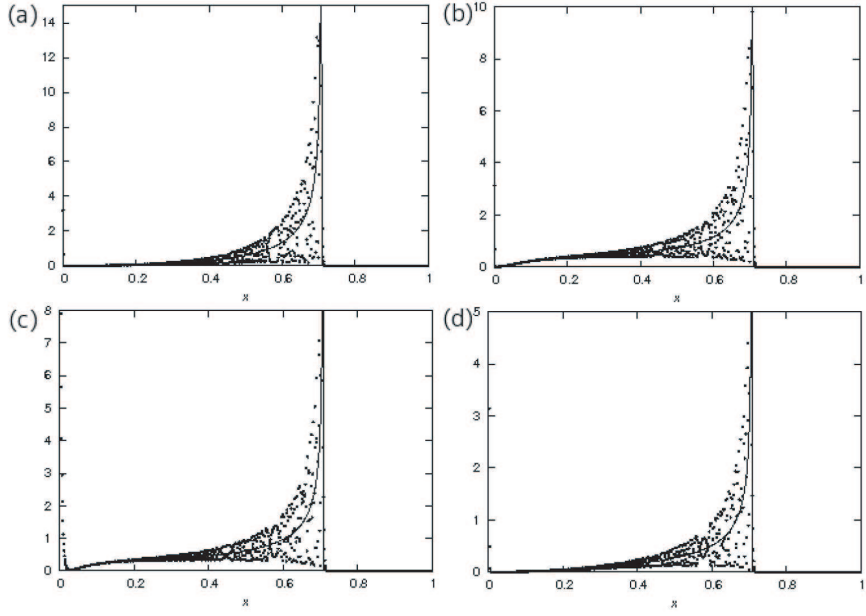}
\end{center}
\caption{Comparison between theoretical values with line and numerical estimations with dots of $W_{t,3}$ with $C=e^{i\varphi}H$.
Density function  $C_d^0(x)f_K(x)$ and scaled numerical values at time $2000$ are plotted. 
In theoretical values, the delta measure is omitted. Since localization has an exponentially decay, 
the scaled numerical values corresponding to localization converge to the delta measure at the infinite time. 
The initial state is $\psi_0=e^{i10\pi/180}/\sqrt{3}$, $\psi_1=e^{i30\pi/180}/\sqrt{3}$, $\psi_2=e^{i340\pi/180}/\sqrt{3}$ 
and (a) $\varphi=0$, (b) $\varphi=40\pi/180$, (c) $\varphi=50\pi/180$, (d) $\varphi=80\pi/180$. 
}
\label{fig:weak convergence}
\end{figure}
%%%%%%%%
Same as other one-dimensional cases \cite{konno01,konno02,grimmett01,miyazaki01,segawa01}, 
this distribution has scaling order $t$ and the typical density function of quantum walks $f_K(x)$. 
The delta measures $C_m$ and $C^r_p$ are caused by localization, i.e., $C_m =\sum_x L_m(x)$ and $C^r_p =\sum_x L^r_p(x)$. 

The above expressions of both theorems seem to be complicated, however for the following
cases, they are written in simpler forms.
%%%%%%%%%%%%%%%%%%%%%%%
\begin{corollary}[$\phi=0$] \label{cor:phi=0}
For $x\in\bZ_+\cup\{0\}$, $r\in\Kk$ and $\phi = 0$, 
\begin{multline*}
P(\Xtr=x) \sim \\ \left( \frac{1+(-1)^{t+x}}{2} \right) \frac{b_\kappa^2 |c|^2}{1+|c|^2}  \left\{ \delta_0(x) 
+ (1-\delta_0(x)) \left(\frac{2}{1+|c|}\right) \left(\frac{1-|c|}{1+|c|} \right)^{x-1} \right\}  \left|\sum_{j\in\mathbb{K}_\kappa}(\psi_j-\psi_r)\right|^2 ,
\end{multline*}
and $\Xtr/t$ converges weakly to a limit measure $\rho_{W_{(0)}}^r$ as $t\to \infty$, where
\begin{multline*}
\rho_{W_{(0)}}^r(x) = \\
\frac{b_\kappa^2|c|}{2(1+|c|)} \left| \sum_{j\in \mathbb{K}_\kappa}(\psi_j-\psi_r)\right|^2 \delta_0(x) 
	+ \frac{1}{|a|^2}\left(\left| (|c|-1)\psi_r + b_\kappa\sum_{j\in \mathbb{K}_\kappa} \psi_j \right|^2 + |a|^2|\psi_r|^2 \right) x^2 f_K(x) .
\end{multline*}
\end{corollary}
%%%%%%%%%%%%%%%%%%%%
\begin{proof}
{From Theorem \ref{thm:localization}, when $\phi = 0$, we have $P(\Xtr=x)\sim(1+(-1)^{t+x})L_p^r(x)/2$.
Also Theorem \ref{thm:weak convergence} implies
$\Gamma_1(x) = |a|^2(a_\kappa^2 + 2a_\kappa|c| + 1)$, 
$\Gamma_2(x) = |a|^2b_\kappa (a_\kappa+|c|)$, 
$\Gamma_3(x) = |a|^2b_\kappa^2$, 
and $K_+K_-=|a|^4$. 
After some calculations of $\theta^r_1(\bi{\psi}),\ \theta^r_2(\bi{\psi}),\ \theta^r_3(\bi{\psi})$, we have the desired conclusion.} 
\end{proof}
%%%%%%%%%%%%%%%%%%%

From the definition of $W_{t,\kappa}$, $W_{t,1}$ is simply a quantum walk on a half line with a reflecting wall on the origin. 
In the next corollary, we denote $X_{t,0}\equiv X_t$ to assert that the walk is defined on a half line.
\begin{corollary}[Half line] \label{cor:half line}
For $x\in\bZ_+\cup\{0\}$,
\begin{multline*}
P\left(\Xt=x\right) \sim I_{[-1,|c|)}\left(\cos{\phi}\right) \left( \frac{1+(-1)^{t+x}}{2} \right) \\ 
	\times\frac{4 |c|^2(\cos{\phi}-|c|)^2}{K_-^2} \left\{ \delta_0(x)+(1-\delta_0(x))\left(\frac{|a|^2}{K_-}\right)^{x-1}\left(1+\frac{|a|^2}{K_-}\right) \right\} ,
\end{multline*}
and $\Xt/t$ converges weakly to a limit measure $\rho_X$ as $t\to \infty$, where
\begin{eqnarray*}
\rho_X(x) = I_{[-1,|c|)}\left(\cos{\phi}\right)\frac{2|c|(|c|-\cos{\phi})}{K_-}\delta_0(x) + \frac{2(1-|c|\cos{\phi})}{K_--(1-x^2)\sin^2{\phi}} x^2 f_K(x).
\end{eqnarray*}
\end{corollary}
%%%%%%%%%%%%%%%%%%%%
\begin{proof}
{For $\kappa=1$, we have  $b_1=2/\kappa=2$, $a_1=b_1-1=1$, $\theta^0_1(\bi{\psi})=1$ and $\theta^0_2(\bi{\psi})=\theta^0_3(\bi{\psi})=0$.
From Theorem \ref{thm:localization}, we get $L_p^0(x)=0$ and $L_c^0(x,t)=0$, thus we should consider only $L_m^0(x)$ as localization factor of $X_t$.
From Theorem \ref{thm:weak convergence}, we have $C_p^0=0$ and
\begin{eqnarray}
\Gamma_1(x) = 2(1-|c|\cos{\phi})\left\{K_+ - (1-x^2)\sin^2{\phi}\right\} . \label{eq:halflinemeasure}
\end{eqnarray}
Combining $C_m$ and $C_d^r(x)$ with Eq. (\ref{eq:halflinemeasure}) implies $\rho_X(x)$.} 
\end{proof} 
%%%%%%%%%%%%%%

\noindent Remark that we get another proof of Corollary \ref{cor:half line} by considering $\Wt$ with the symmetric initial state.% i.e., $\psi_0 = \cdots = \psi_{\kappa-1} = 1/\sqrt{\kappa}$.

We can adopt Corollary \ref{cor:phi=0} for $\Vt$ with no perturbation, i.e., $\tc=1$.
In $\phi=0$ case, the formula for the $r$th half line is directly expressed by $\psi_0,\psi_1,\ldots,\psi_{\kappa-1}$ instead of 
$\theta_1^r(\bi{\psi}),\theta_2^r(\bi{\psi}),\theta_3^r(\bi{\psi})$.
Both cases of $\phi=0$ and half line, the oscillatory term $L_c^r(x,t)$ appearing in Theorem \ref{thm:localization} vanishes.

%%%%%%%%%%%%%%%%%%%%
%%%%%%%%%%%%%%%%%%%%%%%%%%%%%%%%%%%%%%%%%%%%%%%%%%%%%%%%%%%%%%%%%%%%%%%%%%%%%%%%%
%                          section4 proof of theorems
%%%%%%%%%%%%%%%%%%%%%%%%%%%%%%%%%%%%%%%%%%%%%%%%%%%%%%%%%%%%%%%%%%%%%%%%%%%%%%%%%
\section{Proofs of Theorems \ref{thm:localization} and \ref{thm:weak convergence}} \label{sec:proofs of theorems}

In order to prove Theorems \ref{thm:localization} and \ref{thm:weak convergence}, we consider a reduction of $\Wt$ on a half line.
For $\Wt$ with arbitrary initial states, we can not construct the reduction of the walk directly,
since the states with the same distance from the origin have different amplitudes.
To solve this problem, we introduce $\Wtd$ which is a quantum walk with an enlarged basis of $\Wt$.
After that, we construct $\Xtd$ as a reduction of $\Wtd$ on a half line.
To analyze $\Xtd$, we give the generating function of the states.
By using it, we obtain the limit states and the characteristic function of $\Wt$.

%%%%%%%%%%%%%%%%%%%%%%%%%%%%%%%%%%%%%%%%%%%%%%%%%%%%%%%%%%%%%%%%%%%%%%%%%%%%%%%%%
%                     section4.1 reduction to a half line
%%%%%%%%%%%%%%%%%%%%%%%%%%%%%%%%%%%%%%%%%%%%%%%%%%%%%%%%%%%%%%%%%%%%%%%%%%%%%%%%%
\subsection{Reduction to a half line} \label{sec:reduction to a half line}
Let $\Psi_t(x)$ be the state of the quantum walk $\Wt$ at time $t$ and position $x$.
We denote the initial state $\Psi_0(0)$ as $\boldsymbol{\psi}=\sum_{j\in \mathbb{K}_\kappa}\psi_j|0,\epsilon_j\rangle$. 
Now we rewrite $\boldsymbol{\psi}$ using a new orthogonal basis $\{ |\epsilon_j'\rangle; j\in \mathbb{K}_\kappa \}$ as 
\[ 
\boldsymbol{\psi} 
	=\sum_{j\in\mathbb{K}_\kappa}\left(\sum_{r\in\mathbb{K}_\kappa}\psi_r\langle \epsilon'_r|\otimes I_W\right)|\epsilon_j'\rangle|0,\epsilon_j\rangle
	=\sum_{j\in\mathbb{K}_\kappa}\Lambda(\boldsymbol{\psi})|\epsilon_j'\rangle|0,\epsilon_j\rangle,
\]
where $I_W$ is the identity operator on $\sH^{(\Jk)}$ and we defined as
\begin{eqnarray*}
\Lambda(\bi{\psi}) \lequiv \sum_{r\in\Kk} \psi_r \bra{\epsilon_r'} \otimes I_W.
\end{eqnarray*}
Now let $\sH^\prime$ be a Hilbert space spanned by an orthonormal basis $\{ \ket{\erd{j}};\ j\in\Kk \}$.
Then we define $\Wtd$ as a  quantum walk on $\sH^\prime \otimes \sH^{(\Jk)}$ with the evolution operator 
$U_J^\prime = (I_\kappa\otimes S_J)(I_\kappa\otimes F_J) = I_\kappa \otimes U_J$ 
and the initial state 
$\sum_{j\in\Kk} \ket{\erd{j}}\ket{0,\epsilon_j}$, where $I_\kappa$ is the identity operator on $\sH^\prime$.
Let $l_0\lequiv\{ \er{0},\er{1},\ldots,\er{\kappa-1} \}$ and $l_x\lequiv\{Up,Down\}$ for $x\in V(\Jk)\setminus\{0\}$.
Then the state of quantum walk $\Wtd$ at time $t$ and position $x$ is written as 
$\Psi_t^\prime(x) = \sum_{j\in\Kk, u\in l_x} \alpha_t^\prime(\erd{j},x,u)\ket{\erd{j}}\ket{x,u}$,
where $\alpha^\prime_t(a,b,c)$ is the amplitude of the base $\ket{a,b,c}$ at time $t$.
From the construction we obtain the following lemma.
\noindent \\
%
%lemma enlarging basis
%
\begin{lemma}[Enlarging basis] \label{lem:enlarging basis}
For any $t\geq 0$ and $x \in V(\Jk)$,
\[
\Psi_t(x) = \Lambda(\bi{\psi}) \Psi_t^\prime(x) \nonumber.
\]
\end{lemma}
%
%proof
%
\begin{proof}
{We show the equation by induction with respect to $t$. At $t=0$, by definition of $\Lambda(\bi{\psi})$, it is trivial. For fixed $t\geq 1$, 
we assume $\Psi_t(x) = \Lambda(\bi{\psi}) \Psi_t^\prime(x)$, then for $x\in V(\Jk)$,
\begin{eqnarray*}
U_J^{\prime} \Psi^{\prime}_t(x) &=& (I_\kappa\otimes U_J) \sum_{j\in\Kk, u\in l_x} \alpha_t^\prime(\erd{j},x,u)\ket{\erd{j}}\ket{x,u} \nonumber \\
&=& \sum_{j\in\Kk, u\in l_x} \alpha_t^\prime(\erd{j},x,u)\ket{\erd{j}} \left( U_J \ket{x,u} \right), \label{eq:elb1}
\end{eqnarray*}
\begin{eqnarray*}
U_J\Psi_t(x) 
&=& U_J \Lambda(\bi{\psi}) \Psi_t^\prime(x) \label{eq:elb2} \\
&=& U_J \Lambda(\bi{\psi}) \sum_{j\in\Kk, u\in l_x} \alpha_t^\prime(\erd{j},x,u)\ket{\erd{j}}\ket{x,u} \nonumber \\
&=& U_J \sum_{j\in\Kk , u\in l_x} \psi_j^\prime \alpha_t^\prime(\erd{j},x,u)\ket{x,u} \nonumber \\
&=&  \sum_{j\in\Kk , u\in l_x} \psi_j^\prime \alpha_t^\prime(\erd{j},x,u) U_J \ket{x,u} \nonumber \\
&=& \Lambda(\bi{\psi}) \sum_{j\in\Kk, u\in l_x} \alpha_t^\prime(\erd{j},x,u)\ket{\erd{j}} \left( U_J \ket{x,u} \right) \nonumber \\
&=& \Lambda(\bi{\psi}) U^{\prime}_J \Psi^{\prime}_t(x) \label{eq:elb3}.
\end{eqnarray*}
%In the above equation, Eq. (\ref{eq:elb2}) is obtained by the assumption and Eq. (\ref{eq:elb3}) is introduced by Eq. (\ref{eq:elb1}).
This relation holds for any $x$, so we conclude $\Psi_{t+1}(x) = \Lambda(\bi{\psi}) \Psi_{t+1}^{\prime}(x)$.} 
\end{proof}

For $x\in V(\Jk)$, we define the probability of ``$\Wtd = x$'' by
$P(\Wtd=x) \lequiv \| \Lambda(\bi{\psi})\Psi_t^{\prime}(x) \|^2$. 
Then it follows form Lemma \ref{lem:enlarging basis} that $P(\Wtd=x) = P(\Wt=x)$ for any $x\in V(\Jk)$.
%This means it is sufficient to analyze $\Wt$ that we analyze $\Wtd$.
For $\Wtd$, the information of the initial state is covered by $\Lambda(\bi{\psi})$.
In other words, for \textit{any} initial state of $\Wt$,
it is enough to consider the initial state $\sum_{j\in\Kk}\ket{\erd{j}}\ket{0,\epsilon_j}$ on $\Wtd$.
Consequently, the states of the quantum walk $\Wtd$ have a good symmetry, so we can treat the reduction of the walk.

Now we introduce $\Xtd$ as a reduction of $\Wtd$ on a half line.
Here $\Xtd$ is defined on a Hilbert space generated by the following new basis.
For all $l\in\{Up,Down\}$ and $x\in\bZ_+$,
\begin{eqnarray*}
&& \ket{Own,0,\epsilon} \lequiv \sum_{j\in\Kk} \ket{\erd{j},0,\er{j}},\\
&& \ket{Other,0,\epsilon} \lequiv \frac{1}{\sqrt{\kappa-1}}\sum_{j\in\Kk} \sum_{k\in\Kk\setminus\{j\}} \ket{\erd{k},0,\er{j}},\\
&& \ket{Own,x,l} \lequiv \sum_{j\in\Kk}\ket{\erd{j},\hr{j}(x),l}, \\
&& \ket{Other,x,l} \lequiv \frac{1}{\sqrt{\kappa-1}}\sum_{j\in\Kk} \sum_{k\in\Kk\setminus\{j\}} \ket{\erd{k},\hr{j}(x),l}.
\end{eqnarray*}
On this basis, we obtain the one-step time evolution. For $x\in\bZ_+$,
\begin{eqnarray*}
U_J^{\prime}\ :\
\begin{array}{l}
\ket{Own,0,\epsilon} \to a_\kappa \ket{Own,1,Down} +\sqrt{\kappa-1} b_\kappa \ket{Other,1,Down},\\
\ket{Other,0,\epsilon} \to \sqrt{\kappa-1}b_\kappa \ket{Own,1,Down} - a_\kappa \ket{Other,1,Down},\\
\ket{Own,x,Up} \to a\ket{Own,x-1,Up} + c\ket{Own,x+1,Down},\\
\ket{Other,x,Up} \to a\ket{Other,x-1,Up} + c\ket{Other,x+1,Down},\\
\ket{Own,x,Down} \to b\ket{Own,x-1,Up} + d\ket{Own,x+1,Down},\\
\ket{Other,x,Down} \to b\ket{Other,x-1,Up} + d\ket{Other,x+1,Down}.
\end{array}
\end{eqnarray*}
The subspace generated by this basis is invariant under the operation $U_J^\prime$.
Moreover the initial state of $\Wtd$ can be written as $\ket{Own,0,\epsilon}$.
Therefore we can write the evolution operator of $\Xtd$ as $U_H^*=F_H^* S_H^*$. The coin operator $F_H^*$ is defined by
\begin{eqnarray*}
&& F_H^* \lequiv 
\left[ \begin{array}{cc}
a_\kappa & \sqrt{\kappa-1}b_\kappa\\
\sqrt{\kappa-1}b_\kappa &-a_\kappa\\
\end{array} \right]
\otimes \ket{0}\bra{0} \otimes 1 + \sum_{x\in \bZ_+} I_2 \otimes \ket{x}\bra{x} \otimes C. 
\end{eqnarray*}
For $m\in \{Own,Other\}$, $l\in\{Up,Down\}$, the shift operator $S^*$ is defined by
\begin{eqnarray*}
&& S_H^*\ket{m,0,\epsilon} = \ket{m,1,Down}  ,\\
&& S_H^*\ket{m,1,l} = \left\{ 
\begin{array}{ll}
\ket{m,0,\epsilon} , & l = Up ,\\
\ket{m,2,Down} , & l = Down ,\\
\end{array} \right.  \\
&& S_H^*\ket{m,x,l} = \left\{ 
\begin{array}{ll}
\ket{m,x-1,Up} , & l = Up ,\\
\ket{m,x+1,Down} , & l = Down ,\\
\end{array} \right.   \ x\ge 2 .
\end{eqnarray*}
Throughout this paper, we put $\ket{Own}={}^T[1,0]$ and $\ket{Other}={}^T[0,1]$.
An expression of $\Xtd$ using weights is shown by Fig. \ref{fig:hlmat3} and Eqs. (\ref{eq:Q0})-(\ref{eq:IQ}) in Appendix. 

Let $\Psi_t^*(x)$ be the state of the quantum walk $\Xtd$.
Now we define for $r\in\Kk$,
\begin{eqnarray*}
\Lambda_r(\bi{\psi}) \lequiv
\left( \psi_r \bra{Own} + \sqrt{\kappa-1}\sum_{j\in\Kk\setminus\{r\}} \psi_j\bra{Other} \right)\otimes I_W.
\end{eqnarray*}
Then we introduce $\Xtr$ whose probability of ``$\Xtr=x$'' is defined by
\begin{eqnarray}
P(\Xtr=x) \lequiv  \| \Lambda_r(\bi{\psi})\Psi_t^*(x) \|^2. \label{eq:probxtr}
\end{eqnarray}
This probability is described by the state of $W_{t,\kappa}$ in the following, 
\begin{eqnarray*}
P(\Xtr=x) =
\left\{
\begin{array}{ll}
| \alpha_t(0,\er{r}) |^2   ,& x=0 ,\\
\| \Psi_t(\hr{r}(x)) \|^2  ,& \mathrm{otherwise.}\\
\end{array}
\right. 
\end{eqnarray*}
Hence the relation between the probabilities of ``$\Wt=\hr{r}(x)$'' and ``$\Xtr=x$'' is obtained as
\begin{eqnarray*}
P(\Wt=\hr{r}(x)) = 
\left\{ \begin{array}{ll} 
\sum\limits_{j\in\Kk}P(X_{t,j}=0) ,& x=0 ,\\
P(\Xtr=x) ,& \mathrm{otherwise.} \\
\end{array} \right. 
\end{eqnarray*}
Note that $\sum_{j\in\Kk} \sum_{x\in\{0\}\cup\bZ_+}P(X_{t,j}=x) = 1 $.

In Subsections \ref{sec:proof thm1} and \ref{sec:proof thm2}, we analyze $\Psi_t^*(x)$ by the generating function.
%Then from Eq. (\ref{eq:probxtr}), we get Theorems \ref{thm:localization} and \ref{thm:weak convergence}.

%%%%%%%%%%%%%%%%%%%%%%%%%%%%%%%%%%%%%%%%%%%%%%%%%%%%%%%%%%%%%%%%%%%%%%%%%%%%%%%%%
%                        section4.2 proof of theorem1
%%%%%%%%%%%%%%%%%%%%%%%%%%%%%%%%%%%%%%%%%%%%%%%%%%%%%%%%%%%%%%%%%%%%%%%%%%%%%%%%%
\subsection{Proof of Theorem \ref{thm:localization}} \label{sec:proof thm1}
We compute the limit state of $\Xtd$ from the generating function which is defined by
\begin{eqnarray*}
\tilde{\Psi}^*(x;z) \lequiv \sum_{t=0}^\infty \Psi_t^*(x)z^t 
\requiv \sum_{ \substack{l\in \{Own,Other\} \\ m\in\{Up,Down\} } }\tilde{\alpha}^*(l,x,m;z)\ket{l,x,m} . 
\end{eqnarray*}
From Appendix, we see that there exists $0<r_1<1$ so that for any $z$ with $|z|<r_1$,
\begin{eqnarray}
&&\tilde{\alpha}^*(Own,x,Up;z) = 
\left\{ \begin{array}{ll} 
-\frac{d}{ac}(\lambda(z)-az)(\mu(z)+a_\kappa) \Phi(x;z) , & x > 0 ,\\ 
-(\mu(z)+a_\kappa)\mu(z) \Phi(x;z) , & x=0 ,\end{array} \right.  \label{eq:gen1} \\
&&\tilde{\alpha}^*(Other,x,Up;z) = 
\left\{ \begin{array}{ll} 
-\frac{d\sqrt{\kappa-1}}{ac}(\lambda(z)-az)b_\kappa \Phi(x;z) , & x > 0,\\ 
-\sqrt{\kappa-1}b_\kappa\mu(z) \Phi(x;z) , & x=0 ,\end{array} \right.  \label{eq:gen2} \\
&&\tilde{\alpha}^*(Own,x,Down;z) = 
\left\{ \begin{array}{ll} 
-z(\mu(z)+a_\kappa) \Phi(x;z) , & x > 0 ,\\ 
0 , & x=0 ,\end{array} \right.  \label{eq:gen3} \\
&&\tilde{\alpha}^*(Other,x,Down;z) = 
\left\{ \begin{array}{ll} 
-z\sqrt{\kappa-1}b_\kappa \Phi(x;z) , & x > 0 ,\\ 
0 , & x=0 ,\end{array} \right.  \label{eq:gen4} \\
&&\Phi(x;z) = \left\{\frac{d\lambda(z)}{a}\right\}^{x-1} \frac{ w_+^2 w_-^2 \left(\eta_+(z)+\sqrt{\nu(z)} \right)\left(\eta_-(z)-\sqrt{\nu(z)} \right) }{4(1-c^2) (z^2-w_+^2)(z^2-w_-^2) } \label{eq:gen5} ,
\end{eqnarray}
where
\begin{eqnarray*}
&& \lambda(z) \lequiv \frac{\Delta z^2 + 1 - \sqrt{\Delta^2 z^4 + 2\Delta(1-2|a|^2)z^2 + 1}}{2dz} ,\\
&& \mu(z) \lequiv \frac{d\lambda(z)-\Delta z}{c}z ,\\
&& \nu(z) \lequiv (1+\Delta z^2)^2 - 4\Delta |a|^2 z^2 ,\\
&& \eta_\pm(z) \lequiv 2c \pm 1 \mp \Delta z^2 ,\\
&& w_\pm^2 \lequiv \mp \frac{c(1 \pm c)}{\Delta (|a|^2 - 1 \mp c)} ,\\
&& \Delta \lequiv ad-bc .
\end{eqnarray*}
Note that $|w_\pm^2| = \left| (1\pm c)/(\overline{1\pm c}) \right| = 1$.
From Cauchy's theorem, we have for $0<r<r_1<1$,
\begin{equation*}
\Psi_t^*(x) = \frac{1}{2\pi i} \oint_{|z|=r} \tilde{\Psi}^*(x;z)\frac{dz}{z^{t+1}}.
\end{equation*}
Therefore as $t\to \infty$
\begin{multline*}
-\Psi_t^*(x) \sim  {\rm Res}(\tilde{\Psi^*}(x;z),w_+)w_+^{-(t+1)} + {\rm Res}(\tilde{\Psi^*}(x;z),-w_+)(-w_+)^{-(t+1)} \\
+ {\rm Res}(\tilde{\Psi^*}(x;z),w_-)w_-^{-(t+1)} + {\rm Res}(\tilde{\Psi^*}(x;z),-w_-)(-w_-)^{-(t+1)},
\end{multline*}
where ${\rm Res}(f(z),w)$ is the residue of $f(z)$ for $z=w$.
Taking the residues of the generating function, we can compute $\Psi_t^* (x)$. After some calculations with Eq.(9) and $\Psi^*_t(x)$, 
the proof of Theorem 1 is complete. 

\subsection{Proof of Theorem \ref{thm:weak convergence}} \label{sec:proof thm2}
In order to prove Theorem \ref{thm:weak convergence},  
we calculate the Fourier transform of the generating function as 
$\hat{\tilde{\Psi}}^*(s;z)=\sum_x\tilde{\Psi}^*(x;z)e^{isx}$ by Eqs. (\ref{eq:gen1})-(\ref{eq:gen4}).
Then we obtain the characteristic function from the following relation
\begin{eqnarray}
E\left[e^{i\xi \Xtr} \right] 
&=& \sum_{x\in \bZ} \langle \Lambda_r(\bi{\psi})\Psi_t^*(x) , \Lambda_r(\bi{\psi})\Psi_t^*(x) \rangle e^{i\xi x} \nonumber \\
&=& \sum_{x,y\in \bZ} \langle \Lambda_r(\bi{\psi})\Psi_t^*(x) , \Lambda_r(\bi{\psi})\Psi_t^*(y) \rangle e^{i\xi x} \int_0^{2\pi} e^{ik(x-y)}\frac{dk}{2\pi} \nonumber \\
&=& \int_0^{2\pi} \left\{ \sum_{x,y\in \bZ} \langle \Lambda_r(\bi{\psi})\Psi_t^*(x) , \Lambda_r(\bi{\psi})\Psi_t^*(y) \rangle e^{ik(x-y)}e^{i\xi x} \right\} \frac{dk}{2\pi} \nonumber \\
&=& \int_0^{2\pi} \langle \Lambda_r(\bi{\psi})\hat{\Psi}_t^*(s) , \Lambda_r(\bi{\psi})\hat{\Psi}_t^*(s+\xi) \rangle \frac{ds}{2\pi} \label{eq:chara},
\end{eqnarray}
where $\langle\bi{u},\bi{v}\rangle$ is the inner product of vectors $\bi{u}$ and $\bi{v}$.
%where we define the probability of ``$\Xtr=x$'' as
%\begin{eqnarray*}
%P(\Xtr=x) =
%\left\{
%\begin{array}{ll}
%| \alpha_t(0,\er{r}) |^2   ,& x=0 \\
%\| \Psi_t(\er{r}(x)) \|^2  ,& otherwise \\
%\end{array}
%\right. .
%\end{eqnarray*}
%note that $P(\Xtr=x) = \| \Lambda_r^\prime(\bi{\psi})\Psi_t^*(x) \|^2$.

%At first we calculate the Fourier transform of the generating function, $\hat{\tilde{\Psi}}^*(s;z)$.
Now we write the Fourier transform of the generating function as 
\[ \hat{\tilde{\Psi}}^*(s;z)=\sum_{\scriptsize{\begin{matrix}l\in \{Own,Other\} \\ m\in\{Up,Down\}\end{matrix}}}
\hat{\tilde{\alpha}}^*(l,m;s;z)\ket{l,m}. \]
From Eqs. (\ref{eq:gen1})-(\ref{eq:gen4}), we have $\hat{\tilde{\Psi}}^*(s;z)$ as 
\begin{eqnarray*}
&&\hat{\tilde{\alpha}}^*(Own,Up;s;z) = \left(-\mu(z) + \frac{d}{ac}(\lambda(z)-az)\Phi_2(s;z) \right) (a_\kappa+ \mu(z)) \Phi_1(s;z) , \label{eq:fourier1}\\
&&\hat{\tilde{\alpha}}^*(Own,Down;s;z) = \left(-\mu(z) + \frac{d}{ac}(\lambda(z)-az)\Phi_2(s;z) \right) \sqrt{\kappa-1} b_\kappa \Phi_1(s;z) ,\label{eq:fourier2}\\ 
&&\hat{\tilde{\alpha}}^*(Other,Up;s;z) =  z( a_\kappa + \mu(z) )\Phi_1(s;z) \Phi_2(s;z),\label{eq:fourier3}\\
&&\hat{\tilde{\alpha}}^*(Other,Down;s;z) = z \sqrt{\kappa-1} b_\kappa \Phi_1(s;z) \Phi_2(s;z), \label{eq:fourier4}
\end{eqnarray*}
where
\begin{eqnarray*}
&& \Phi_1(s;z) \lequiv \frac{w_+^2 w_-^2}{4(1-c^2)} \frac{ \left(\eta_+(z)+\sqrt{\nu(z)} \right) \left(\eta_-(z)-\sqrt{\nu(z)} \right) }{(z^2-w_+^2)(z^2-w_-^2)} , \\
&& \Phi_2(s;z) \lequiv \frac{ e^{ik}\left( \zeta(s;z)-\sqrt{\nu(z)} \right) }{2\Delta(z-v_+(s))(z-v_-(s))} ,\\
&& \zeta(s;z) \lequiv 2ae^{-is}z - 1 - \Delta z^2 , \\
&& v_\pm(s) \lequiv \frac{ae^{-is}+\overline{a}\Delta e^{is} \pm \sqrt{(ae^{-is}+\overline{a}\Delta e^{is})^2 - 4\Delta } }{2\Delta} .
\end{eqnarray*}
Here we can rewrite $v_\pm(s)$ as
\begin{eqnarray*}
v_\pm(s) = e^{-i\rho}(|a|\cos{\gamma(s)} \pm \sqrt{|a|^2 \cos^2{\gamma(s)}-1}) = e^{-i\rho}e^{\pm i\theta(s)},
\end{eqnarray*}
where we take $\Delta=e^{2i\rho}$, $a=|a|e^{i\sigma}$, $\gamma(s) = s-\sigma + \rho$, $\cos{\theta(s)} = |a|\cos{\gamma(s)}$.
Note that $|v_\pm(s)|=1$. Now $||\hat{\tilde{\Psi}}^* (s;z)||^2<\infty$ for $0<|z|<r_1$, 
we can rewrite $\hat{\tilde{\Psi}}^*(s;z)=\sum_{t\geq 0}\hat{\Psi}^*_t(s)z^t$. 
So we have for $0<r<r_1$
\begin{eqnarray*}
\hat{\Psi}_t^*(s) = \frac{1}{2\pi i}\oint_{|z|=r}\hat{\tilde{\Psi}}(s;z)\frac{dz}{z^{t+1}} \label{eq:fourier} .
\end{eqnarray*}
Therefore we get the Fourier transform of the state $\hat{\Psi}_t^*(s)$ as follows:
\begin{multline}
-\hat{\Psi}_t^*(s) \sim  \psi_{w_+}(s)(w_+)^{-(t+1)} + \psi_{-w_+}(s)(-w_+)^{-(t+1)} \\
\hspace{15mm} + \psi_{w_-}(s)(w_-)^{-(t+1)} + \psi_{-w_-}(s)(-w_-)^{-(t+1)} \\
\hspace{15mm} + \psi_{v_+}(s) (v_+(s))^{-(t+1)} + \psi_{v_-}(s) (v_-(s))^{-(t+1)}  , \label{eq:res}
\end{multline}
where $\psi_{\pm w_\pm}(s) = {\rm Res}(\hat{\tilde{\Psi}}(s;z);\pm w_\pm)$ and $\psi_{v_\pm}(s) = {\rm Res}(\hat{\tilde{\Psi}}(s;z);v_\pm(s))$.

Finally we compute the characteristic function by Eqs.(\ref{eq:chara}) and (\ref{eq:res}). Now we have 
\begin{eqnarray*}
&& \int_0^{2\pi} \left(\|\Lambda_r(\bi{\psi}) \psi_{w_+}(s)\|^2 + \|\Lambda_r(\bi{\psi}) \psi_{-w_+}(s)\|^2 \right) \frac{ds}{2\pi} = C_p^r , \label{eq:chara1} \\
&& \int_0^{2\pi} \left(\|\Lambda_r(\bi{\psi}) \psi_{w_-}(s)\|^2 + \|\Lambda_r(\bi{\psi}) \psi_{-w_-}(s)\|^2 \right) \frac{ds}{2\pi} =C_m , \label{eq:chara2}\\
&& \int_0^{2\pi} \left( \langle \Lambda_r(\bi{\psi})\psi_{w_+}(s) , \Lambda_r(\bi{\psi})\psi_{-w_+}(s) \rangle 
+ \langle \Lambda_r(\bi{\psi})\psi_{-w_+}(s) , \Lambda_r(\bi{\psi})\psi_{w_+}(s) \rangle \right) \frac{ds}{2\pi} = 0 , \label{eq:chara3}\\
&& \int_0^{2\pi} \left( \langle \Lambda_r(\bi{\psi})\psi_{w_-}(s) , \Lambda_r(\bi{\psi})\psi_{-w_-}(s) \rangle
+ \langle \Lambda_r(\bi{\psi})\psi_{-w_-}(s) , \Lambda_r(\bi{\psi})\psi_{w_-}(s) \rangle \right) \frac{ds}{2\pi} = 0 . \label{eq:chara4}
\end{eqnarray*}
Noting that 
$e^{i(t+1)\theta(s+\xi/t)}=e^{i(t+1)\theta(s)+i\xi h(s)+o(t^{-1})}$, where $h(s)\lequiv d\theta(s)/ds$, 
the above equation and Eq.(\ref{eq:chara}) with the Riemann-Lebesgue lemma imply
\begin{equation}
\lim_{t\to \infty} E\left[e^{i\xi \Xtr/t} \right] = C_p^r + C_m + \int_0^{2\pi} e^{-i\xi h(s)} p(s)\frac{ds}{2\pi} + \int_0^{2\pi} e^{i\xi h(s)} q(s)\frac{ds}{2\pi}, \label{eq:chara5}
\end{equation}
where 
$p(s)=\| \Lambda_r(\bi{\psi})\psi_{v_+}(s)\|^2$ and $q(s)=\| \Lambda_r(\bi{\psi})\psi_{v_-}(s)\|^2$. 
Moreover, from a change of variable for last two terms in Eq. (\ref{eq:chara5}), we have 
\begin{equation*}
\int_0^{2\pi} \left( e^{-i\xi h(s)} p(s) + e^{i\xi h(s)} q(s) \right) \frac{ds}{2\pi} 
=\int_0^\infty e^{i\xi x} w(x) f_K(x) dx.
\end{equation*}
After some calculations for $w(x)$ with $p(s)$ and $q(s)$, we have the desired conclusion.

%%%%%%%%%%%%%%%%%%%%%%%%%%%%%%%%%%%%%%%%%%%%%%%%%%%%%%%%%%%%%%%%%%%%%%%%%%%%%%%%%
%                        section5 summary and disussions
%%%%%%%%%%%%%%%%%%%%%%%%%%%%%%%%%%%%%%%%%%%%%%%%%%%%%%%%%%%%%%%%%%%%%%%%%%%%%%%%%
\section{Summary and discussions} \label{sec:summary}
We introduced a quantum walk with an enlarged basis to consider a reduction of quantum walks with arbitrary initial state.
This method is based on an idea canceling the asymmetry caused from initial state by a new tensor product.
%When we consider quantum walks starting at one point,
%the idea seems to be adopted to some other quantum walks on graphs which can be reduced on simpler graphs with the \textit{symmetric} initial state \cite{krovi01}, 
%even if the initial state is \textit{asymmetric}.
From our results in this paper, we discuss two interesting points.
First, we found the oscillating probability as localization.
From Theorem \ref{thm:localization}, the oscillatory term is expressed by $L_c^r(x,t)$.
We can see that this term vanishes with some initial states or local coins.
For example, we consider $\Vt$, which is a quantum walk on $\bT_\kappa$ with local coin operators $G_\kappa$ with additional complex phase $\tc$ at the origin.
If $\tc=1$ Corollary \ref{cor:phi=0} implies that the oscillation does not occur with arbitrary initial state.
Also if the initial state is symmetric, the walk is reduced on a half line. 
Then it follows from Corollary \ref{cor:half line} that no oscillatory behavior arise with any complex phase $\tc$.
%From existing works \cite{inui01,inui02,inui03,chisaki01}, localizations caused by homogeneous coined quantum walk do not have oscillating terms.
Thus the initial state and differences on complex phase of local coins are important factors for the oscillatory behavior on localization.
Especially in quantum walks on the one-dimensional lattice with homogeneous local coins, localization does not occur \cite{konno01,konno02,grimmett01}.
If localization occurs with perturbations of local coin operators on the one-dimensional lattice,
there seems to be a condition that an oscillatory behavior arises in localization.
Second, $\Wt$ has the scaling order $t$ and the limit measure has the density function $f_K(x)$ 
which is a half-line version of one appearing in the quantum walk on a line \cite{konno01,konno02,grimmett01}.
This is a typical property of quantum walks \cite{chisaki01,miyazaki01,segawa01}.
%Our result also indicates a universality of such a type of density function on quantum walks.
To show the universality of the limit theorems for quantum walks is one of the interesting future's problems.

\par
\
\par\noindent
\noindent
{\bf Acknowledgments.}
\noindent We thank Noriko Saitoh and Jun Kodama for useful discussions. N.K. is supported by the
Grant-in-Aid for Scientific Research (C) (No. 21540118). 

\par
\
\par

\begin{small}
\bibliographystyle{jplain}

\end{small}

\noindent\\
\noindent\\
\noindent\\

\appendix 
\section*{Appendix}
\begin{figure}[h]
\centering
\resizebox{0.4\textwidth}{!}{%
  \includegraphics{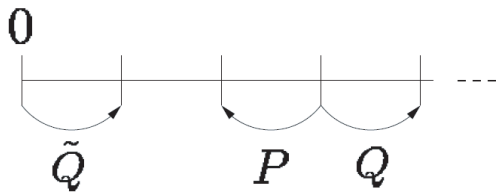}
}
%\vspace{5cm}       % Give the correct figure height in cm
\caption{Quantum walk with enlarged bases $\Xtd$ on a half line}
\label{fig:hlmat3}       % Give a unique label
\end{figure}

We calculate the generating function of $\Psi_t^*(x)$ by using the method in \cite{oka01}.
To simplify notations, for $l\in \{Own,Other\}$, we denote $|l,0,\epsilon\rangle=|l,0,Up\rangle$ and construct $|l,0,Down\rangle$ 
as a dummy base, which always has value $0$ as its amplitude, so that the local coin operator on the origin has $2\times 2$ matrix. 
To indicate the evolution operator of the walk, we use an expression using weights (see Fig.\ref{fig:hlmat3}), where 
\begin{eqnarray}
&& \tilde{Q} \lequiv
\left[ \begin{array}{cc} a_\kappa& \sqrt{\kappa-1}b_\kappa\\ \sqrt{\kappa-1}b_\kappa&-a_\kappa\\ \end{array} \right]
\otimes
\left[ \begin{array}{cc} 0&0\\ 1&0\\ \end{array} \right] , \label{eq:Q0}\\ 
&& P =
I_2
\otimes
\left[ \begin{array}{cc} a&b\\ 0&0\\ \end{array} \right] \ ,\ 
Q =
I_2
\otimes
\left[ \begin{array}{cc} 0&0\\ c&d\\ \end{array} \right] , \label{eq:PQ}\\
&& \Psi_0^*(0) = 
\left[ \begin{array}{c} 1\\ 0\\ \end{array} \right] 
\otimes
\left[ \begin{array}{c} 1\\ 0\\ \end{array} \right] . \label{eq:IQ}
\end{eqnarray}
%where we regard as
%\begin{eqnarray*}
%&& \ket{Own} = \left[ \begin{array}{c} 1\\ 0\\ \end{array} \right] ,\ 
%\ket{Other} = \left[ \begin{array}{c} 0\\ 1\\ \end{array} \right] , \\
%&& \ket{Up} = \left[ \begin{array}{c} 1\\ 0\\ \end{array} \right] , \ 
%\ket{Down} = \left[ \begin{array}{c} 0\\ 1\\ \end{array} \right] , \\
%\end{eqnarray*}
%Note that the right-hand tensor product of $P$, $Q$ and $\tilde{Q}$ are same as construction in \cite{oka01},
%moreover the left-hand tensor product of $P$ and $Q$ are $I_2$ and squared one of $\tilde{Q}$ coresponds to $I_2$,
%so we can compute the generating function by the same method in \cite{oka01}, to divide paths into paths reaching the origin even times and odd times.
We define the generating function for the state by
\begin{eqnarray*}
\tilde{\Psi}^*(x;z) \lequiv \sum_{t=0}^\infty \Psi_t^*(x)z^t .
\end{eqnarray*}
In order to compute $\tilde{\Psi}^*(x;z)$, 
we first define the transition amplitude $\tilde{\Xi}(0\to x;\tau)$ as the weight of all paths starting from $0$ ending at $x$ after $\tau$ steps,
and $\Xi(0\to x;\tau)$ as the weight of all paths on another walk defined by $Q_0^\prime \lequiv Q$.
For example, $\tilde{\Xi}(0\to 2;4)=QPQ\tilde{Q} + PQQ\tilde{Q} + Q\tilde{Q}P\tilde{Q}$ and $\Xi(0\to 2;4)=QPQQ + PQQQ + QQPQ$. 
From $\tilde{\Xi}(0\to 0;\tau_1)$ and $\Xi(0\to x-1;\tau_2)$, we can obtain $\tilde{\Xi}(0\to x;\tau_1+\tau_2+1)$ as Fig.\ref{fig:Xi}.
Then we get $\tilde{\Psi}^*(x;z)$ from the generating function for $\tilde{\Xi}(0\to x;\tau)$.
\begin{figure}[t]
 \centering
 \resizebox{0.55\textwidth}{!}{%
   \includegraphics{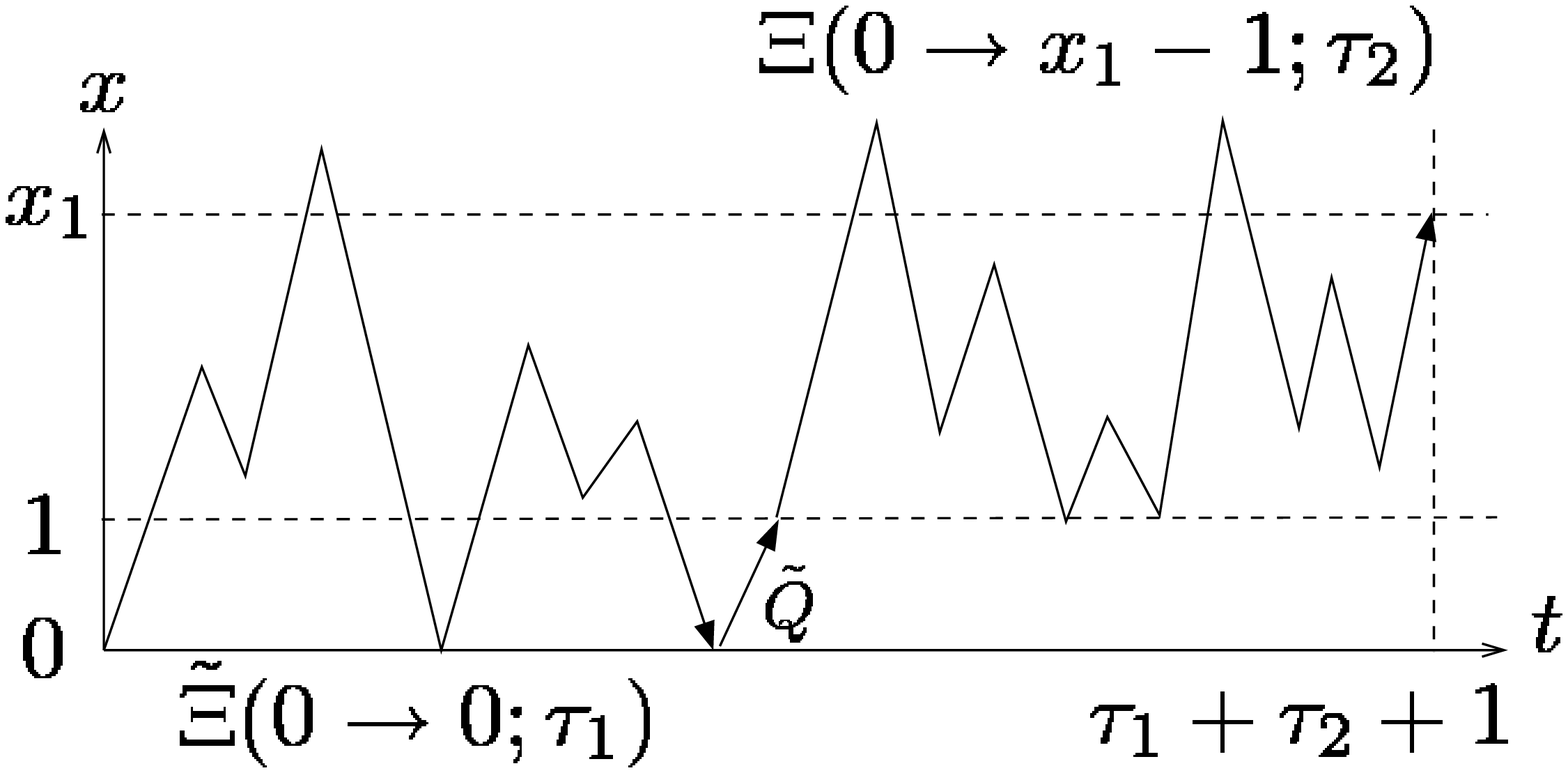}
 }
 %\vspace{5cm}       % Give the correct figure height in cm
 \caption{$\tilde{\Xi}(0\to x_1;\tau_1+\tau_2+1)$}
 \label{fig:Xi}       % Give a unique label
\end{figure}

We now calculate the generating function for $\Xi(0\to x;\tau)$.
Since the first operator should be $Q$ on the half line, the weights of paths form $Q\cdots Q$ or $P\cdots Q$.
So we express $\Xi(0\to x;\tau)$ as a linear combination of $Q$ and $R$:
\begin{eqnarray*}
\Xi(0\to x;\tau)=  b^q(0\to x;\tau)Q + b^r(0\to x;\tau)R  + \delta_0(x)\delta_0(\tau) I_2 \otimes I_2
\end{eqnarray*}
where $R\lequiv I_2 \otimes \left[ \begin{array}{cc} c&d\\ 0&0\\ \end{array} \right] $ and we define $b^q(0\to x;0)=b^r(0\to x;0)=0$.
The generating function for $\Xi(0\to x;\tau)$ is defined by
\begin{eqnarray*}
\sum_{\tau=0}^\infty\Xi(0\to x;\tau)z^\tau=B^q(0\to x;z)Q + B^r(0\to x;z)R + \delta_0(x) I_2 \otimes I_2,
\end{eqnarray*}
with $B^q(0\to x;z) = \sum_{\tau=0}^\infty b^q(0\to x;\tau)z^\tau$ and $B^r(0\to x;z) = \sum_{\tau=0}^\infty b^r(0\to x;\tau)z^\tau$.
Since the left-hand tensor product of $P$ and $Q$ is $I_2$, the generating function for $\Xi(0\to x;\tau)$ corresponds to the result in \cite{oka01}, i.e., 
for sufficiently small $z$,
\begin{eqnarray}
&& B^q(0\to x;z) = \left\{ \frac{d}{a}\lambda(z) \right\}^x \frac{1}{d} ,\  x\geq 1 , \nonumber \\
&& B^q(0\to 0;z) = 0 ,  \nonumber \\
&& B^r(0\to x;z) = \left\{ \frac{d}{a}\lambda(z) \right\}^x \frac{\lambda(z) - az}{acz} ,\  x\geq 0 , \nonumber \\
&& \lambda(z) = 
\frac{\Delta z^2 + 1 - \sqrt{\Delta^2 z^4 + 2\Delta(1-2|a|^2)z^2 + 1}}{2dz}  . \label{eq:lambda}
\end{eqnarray}
Here we take $\lambda(z)$ for the smaller solution of the absolute value of
\begin{eqnarray}
\lambda^2(z) - \frac{1}{d}\left( \Delta z + \frac{1}{z} \right) \lambda(z) + \frac{a}{d} = 0. \label{eq:lambda-master}
\end{eqnarray}
Note that for sufficiently small $z$ we can write $\lambda(z)$ by Eq. (\ref{eq:lambda}).
Moreover since $|a/d|=1$, we can take $r_0 < 1$ such that $|\lambda(z)|<1$ for $|z|<r_0$. 
Next we calculate the generating function for $\tilde{\Xi}(0\to 0;\tau)$. 
To do so, we introduce a new notation $\tilde{\Xi}(0\to 0;\tau;n)$ as the weight of all paths starting from the origin reaching the origin 
$n$ times before ending at the origin at time $\tau$.
Now we consider $\tilde{\Xi}(0\to 0;\tau;0)$.
For $\tau\geq 2$, we obtain $\tilde{\Xi}(0\to 0;\tau;0)$ as  (see Fig.\ref{fig:Xi0})
\begin{figure}[t]
 \centering
 \resizebox{0.45\textwidth}{!}{%
   \includegraphics{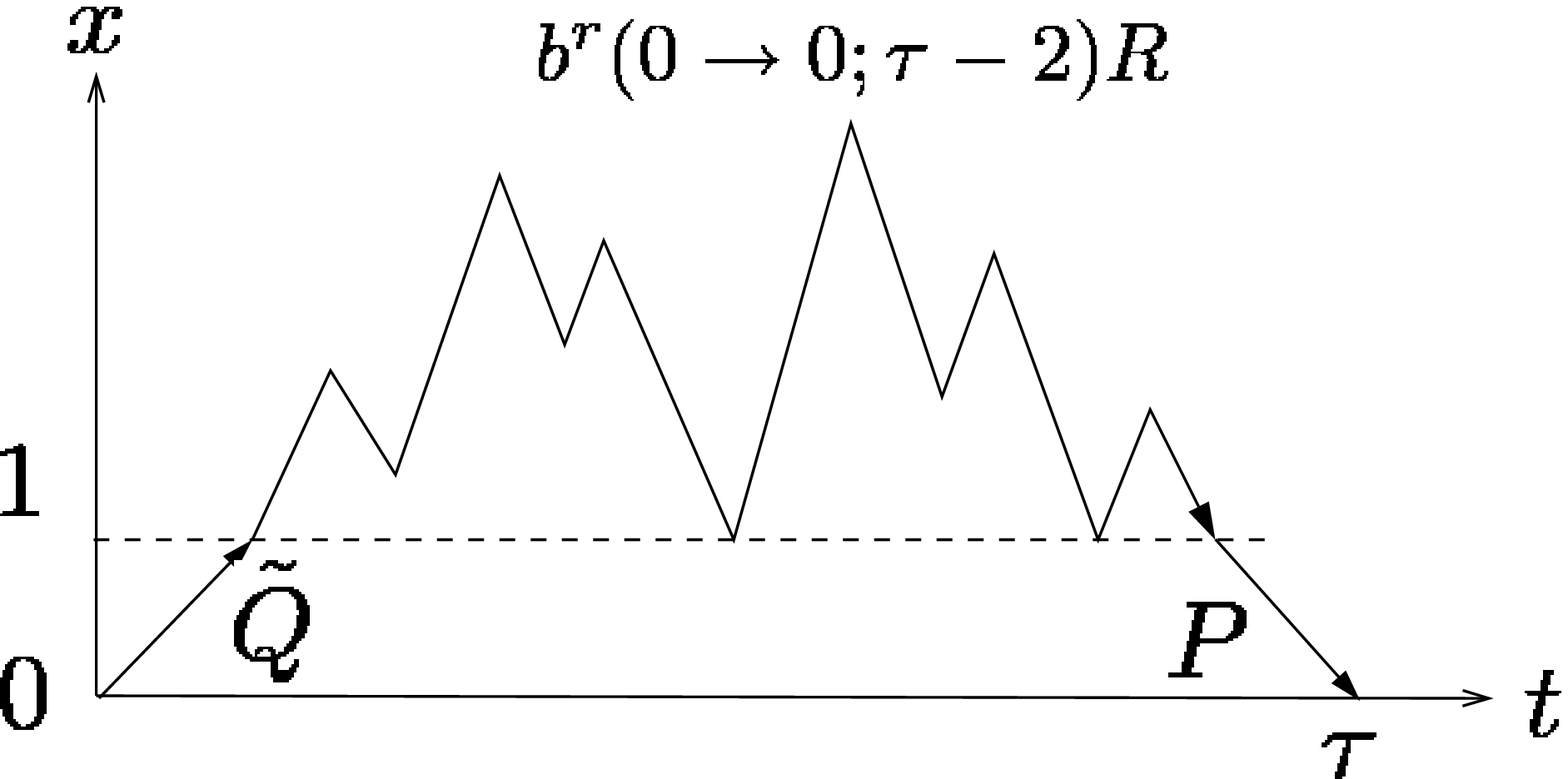}
 }
 %\vspace{5cm}       % Give the correct figure height in cm
 \caption{$\tilde{\Xi}(0\to 0;\tau;0)$}
 \label{fig:Xi0}       % Give a unique label
\end{figure}
\begin{eqnarray*}
\tilde{\Xi}(0\to 0;\tau;0) &=& (1-\delta_2(\tau))P \{ b^r(0\to0;\tau-2)R \} \tilde{Q} + \delta_2(\tau)P\tilde{Q} \\
&=& \{(1-\delta_2(\tau))abb^r(0\to 0;\tau-2)+\delta_2(\tau)b \} \tilde{R},
\end{eqnarray*}
where $\tilde{R}\lequiv\left[ \begin{array}{cc} a_\kappa&\sqrt{\kappa-1}b_\kappa\\ \sqrt{\kappa-1}b_\kappa&-a_\kappa\\ \end{array} \right] 
	\otimes \left[ \begin{array}{cc} 1 &0\\ 0&0\\ \end{array} \right] $ 
and for $\tau<2$ we define $\tilde{\Xi}(0\to 0;\tau;0)=0$.
Therefore we get the generating function for $\tilde{\Xi}(0\to 0;\tau;0)$ as
\begin{eqnarray*}
\sum_{\tau=0}^\infty \tilde{\Xi}(0\to 0;\tau;0)z^\tau 
= (adB^r(0\to 0;z) + b)z^2\tilde{R} 
= B^{\tilde{r}}(0\to 0;z;0)\tilde{R}.
\end{eqnarray*}
Similarly, for $\tau\geq 4$ we have $\tilde{\Xi}(0\to 0;\tau;1)$ as
\begin{multline*}
\tilde{\Xi}(0\to 0;\tau;1) = \sum_{\tau_1+\tau_2+4=\tau}\{(1-\delta_2(\tau_1))abb^r(0\to 0;\tau_1)+\delta_2(\tau_1)b\}\tilde{R} \\ 
	\times \{(1-\delta_2(\tau_2))abb^r(0\to 0;\tau_2)+\delta_2(\tau_2)b\}\tilde{R},
\end{multline*}
and for $\tau<4$ we define $\tilde{\Xi}(0\to 0;\tau;1)=0$.
Thus the generating function for $\tilde{\Xi}(0\to 0;\tau;1)$ is obtained by
\begin{eqnarray*}
\sum_{\tau=0}^\infty \tilde{\Xi}(0\to 0;\tau;1)z^\tau 
= \{(adB^r(0\to 0;z) + b)z^2\}^2\tilde{R}_I 
= B^{\tilde{r}_I}(0\to 0;z;1)\tilde{R}_I,
\end{eqnarray*}
where $\tilde{R}_I\lequiv I_2 \otimes \left[ \begin{array}{cc} 1&0\\ 0&0\\ \end{array} \right] $. 
Recursively we have the following formulae: for $n\geq 0$,
\begin{eqnarray}
&& B^{\tilde{r}} (0\to 0;z;n) = \left( \frac{1+(-1)^n}{2} \right) \{ (adB^r(0\to 0;z)+b)z^2\}^{n+1}, \label{eq:tr}\\
&& B^{\tilde{r}_I} (0\to 0;z;n) = \left( \frac{1+(-1)^{n+1}}{2} \right) \{(adB^r(0\to 0;z)+b)z^2\}^{n+1}. \label{eq:tri}
\end{eqnarray}
From Eqs. (\ref{eq:tr}) and (\ref{eq:tri}), we get the generation function for $\tilde\Xi(0\to 0;\tau)$ by summing over $n$.
Here $(dB^r(0\to 0;z)-\Delta z)z^2 = (d\lambda(z)-\Delta z)z/c$, so we see that for $z$ with $|z|< r_1 \equiv {\rm min}(|c|,r_0)$,
\begin{eqnarray*}
|(d\lambda(z)-\Delta z)z/c|^2 \leq (|d\lambda(z)|^2 + |z|^2)|z/c|^2 < |d|^2 + |c|^2 = 1.
\end{eqnarray*}
Therefore for $z$ such that $|z|<r_1$,
\begin{eqnarray*}
&& \sum_{\tau=0}^\infty \tilde{\Xi}(0\to 0;\tau)z^\tau = B^{\tilde{r}}(0\to 0;z)\tilde{R} + B^{\tilde{r}_I}(0\to 0;z)\tilde{R}_I + I_2 \otimes I_2, \\
&& B^{\tilde{r}}(0\to 0;z) = \sum_{n=0}^\infty B^{\tilde{r}}(0\to 0;z;n)= \frac{(d\lambda(z)-\Delta z)z/c}{1-\{\tc (d\lambda(z)-\Delta z)z/c \}^2} , \\
&& B^{\tilde{r}_I}(0\to 0;z) = \sum_{n=0}^\infty B^{\tilde{r}_I}(0\to 0;z;n) = \frac{\{(d\lambda(z)-\Delta z)z/c\}^2}{1-\{(d\lambda(z)-\Delta z)z/c \}^2}.
\end{eqnarray*} 
For $x\geq 1$, $\tilde{\Xi}(0\to x;\tau)$ is written by $\tilde{\Xi}(0\to 0;\tau)$ and $\Xi(0\to x;\tau)$ (see Fig.\ref{fig:Xi}) as
\begin{eqnarray*}
\tilde{\Xi}(0\to x;\tau)  = \sum_{\tau_1+\tau_2+1=\tau}\Xi(0\to x-1;\tau_2) \tilde{Q} \tilde{\Xi}(0\to 0;\tau_1) + \delta_0(\tau)\delta_0(x)I_2 \otimes I_2 .
\end{eqnarray*}
From the generating function for $\tilde{\Xi}(0\to 0;\tau)$ and $\Xi(0\to x;\tau)$,
we can compute the generating function for $\tilde{\Xi}(0\to x;\tau)$ as follows: for $x\geq 1$,
\begin{eqnarray*}
&&\sum_{\tau=0}^\infty \tilde{\Xi}(0\to x;\tau)z^\tau 
= \{  B^q(0\to x-1;z)Q + B^r(0\to x-1;z)R + \delta_1(x)I_2\otimes I_2 \} \tilde{Q}z \nonumber \\
&& \hspace{40mm}  \times \{ B^{\tilde{r}}(0\to0;z)\tilde{R} + B^{\tilde{r}_I}(0\to0;z)\tilde{R}_I + I_2\otimes I_2 \} + \delta_0(x)I_2 \otimes I_2.
\end{eqnarray*}
Then we obtain the generating function $\tilde{\Psi}^*(x;z)$ as follows:
\begin{eqnarray*}
\tilde{\Psi}^*(x;z) = \sum_{\tau=0}^\infty \tilde{\Xi}(0\to x;\tau)z^\tau \Psi_0^*(x) .
\end{eqnarray*}

\end{document}